%% file: IH2012_Stat_detect_LSB_match_v3RC.tex
\documentclass[english]{llncs}

\usepackage{ifpdf}
\usepackage{amsmath,amssymb,amsfonts}
\usepackage[T1]{fontenc}
\usepackage[latin9]{inputenc}
\usepackage[font=footnotesize]{subfig}
\usepackage{nicefrac}
\usepackage{color}
\usepackage{url}
  \usepackage[font=footnotesize]{subfig}

\ifpdf
  \usepackage[pdftex]{graphicx}
  \DeclareGraphicsExtensions{.pdf}
\else
  \usepackage[dvips]{graphicx}
  \DeclareGraphicsExtensions{.eps}
\fi
\graphicspath{./images/}

\newcommand{\refeq}[1]{(\ref{#1})}

\newcommand{\prob}{\mathbb{P}}
\newcommand{\expec}{\mathbb{E}}
\newcommand{\var}{\mathbb{V}ar}

\newcommand{\calN}{\mathcal N}

\newcommand{\calK}{\mathcal K}
\newcommand{\calH}{\mathcal H}

\newcommand{\calZ}{\mathcal Z}

\newcommand{\mZ}{\mathbf Z}

\newcommand{\mM}{\mbox{\textit{\textbf{M}}}}
\newcommand{\mC}{\mbox{\textit{\textbf{C}}}}
\newcommand{\mS}{\mbox{\textit{\textbf{S}}}}

\newcommand{\ie}{\textit{i.e. }}

\newcommand{\R}{\mathbb{R}}

\def\bydef{\mathop{\overset{\mbox{\scriptsize def.}}{=}}\nolimits}

\newcommand{\myProject}{This work was partially supported by National Agency for Research (ANR) through ANR-CSOSG Program (Project ANR-07-SECU-004). }
\newcommand{\myProjectb}{With the financial support from the Prevention of and Fight against Crime Programme of the European Union European Commission - Directorate-General Home Affairs. (2centre.eu project).}
\newcommand{\myProjectc}{Research partially funded by Troyes University of Technology (UTT) strategic program COLUMBO. }
\newcommand{\myLab}{ICD - LM2S - Université de Technologie de Troyes - UMR STMR CNRS\\
12, rue Marie Curie - B.P. 2060 - 10010 Troyes cedex - France\\
E-mail : name.surname@utt.fr}

\title{Statistical Detection of LSB Matching\\Using Hypothesis Testing Theory}
\author{Rémi Cogranne, Cathel Zitzmann,  Florent Retraint, Igor Nikiforov, Lionel Fillatre and Philippe Cornu \thanks{\myProject \myProjectb \myProjectc}}
\institute{\myLab}

\begin{document}

\makeatother

\maketitle

\begin{abstract}
This paper investigates the detection of information hidden by the
Least Significant Bit (LSB) matching scheme. In a theoretical context
of known image media parameters, two important results are presented.
First, the use of hypothesis testing theory allows us to design the
Most Powerful (MP) test. Second, a study of the MP test gives us the
opportunity to analytically calculate its statistical performance
in order to warrant a given probability of false-alarm. In practice
when detecting LSB matching, the unknown image parameters have to
be estimated. Based on the local estimator used in the Weighted Stego-image
(WS) detector, a practical test is presented. A numerical comparison
with state-of-the-art detectors shows the good performance of the
proposed tests and highlights the relevance of the proposed methodology. 
\end{abstract}

\section{Introduction and Contributions.}\label{ssec:introduction}
Steganography and steganalysis form a cat-and-mouse game. On the one
hand, steganography aims at hiding the very presence of a secret message
by hiding it within an innocuous cover medium. On the other hand,
the goal of steganalysis (in the wide sense) is to obtain any information
about the potential steganographic system from an unknown medium.
Usually, steganalysis focuses on exposing the existence of a hidden
message in an inspected medium.\\
 Many steganographic tools are nowadays easily available on the
Internet making steganography within the reach of anyone, for legitimate
or malicious usage. It is thus crucial for security forces to be able
to reliably detect steganographic content among a (possibly very large)
set of media files. In this operational context, the detection of
a rather simple but most commonly found stegosystem seems more important
than the detection of a very complex but rarely encountered stegosystem.
The vast majority of downloadable steganographic tools insert the
secret information in the LSB plane. Consequently, substantial progress
has recently been made in the detection of such steganographic algorithms,
namely LSB replacement and LSB matching, also known as LSB $\pm1$
embedding (see~\cite{CoxFri07,Fridrich2010bk,Bohme2010} and the
references therein). However, the steganalysis of LSB matching remains
much harder than the steganalysis of LSB replacement. Indeed, if LSB
matching is used instead of LSB replacement, the detection power of
state-of-the-art detectors is significantly lower~\cite{CoxDoerr07pm,DoerrCox08}.

The recently proposed steganalyzers dedicated to LSB matching can
be roughly divided into two categories. 
On the one hand, most of the latest detectors are based on supervised
machine learning methods and use targeted~\cite{CoxDoerr08pm,pmspeffeatures10}
or universal features~\cite{Goljan06newblind,Lyu06TIFS}. As in all
applications of machine learning, the theoretical calculation of error
probabilities remains an open problem~\cite{scott07npperform}. On
the other hand, the authors of~\cite{Harmsen05pm} observed that
LSB matching acts as a low-pass filter on the image Histogram Characteristic
Function (HCF). This pioneering work lead to an entire family of histogram-based
detectors~\cite{ker05pm,CoxDoerr07pm}.

In the operational context described above, the proposed steganalyzer
must be immediately applicable without any training or tuning phase.
For this reason, the use of a machine learning based detector is hardly
possible. Moreover, the most important challenge for the steganalyst
is to provide detection algorithms with an analytical expression for
the false-alarm and missed-detection probabilities without which the
{}``uncertainty'' of the result can not be {}``measured.'' The
proposed LSB matching steganalyzers are certainly very interesting
and efficient, but these \emph{ad hoc} algorithms have been designed
with a very limited exploitation of statistical cover models and hypothesis
testing theory. Hence, a few theoretical results exist and the only solution to measure their statistical performance is the simulation on large databases.

Alternatively, the first step in the direction of hypothesis testing
has been made in~\cite{DMS2004,Cog2011SSP,Cog2011ISIT} for LSB replacement
to design a statistical test with known statistical properties. In
the present paper, this statistical approach is extended to the case
of detecting LSB matching. More precisely, the goal of this paper
is threefold: 
\begin{enumerate}
\item Define the most powerful (MP) test in the theoretical case when the
cover image parameters are known, namely the expectation and noise
variance of each pixel.
\item Analytically calculate the statistical performance of the MP test
in terms of the false-alarm and missed-detection probabilities. More
importantly, this result allows us to highlight the impact of the
noise variance and quantization on the test performance~\cite{Cog2011ISIT}. 
\item Design a practical efficient implementation of this test based on
a simple local estimation of expectation and variance of each pixel. 
\end{enumerate}

The paper is organized as follows. 
The problem of LSB matching steganalysis is casted within the framework of hypothesis testing in Section~\ref{sec:problem_stat}. Following the Neyman-Pearson approach, the MP Likelihood Ratio Test
(LRT) is presented in Section~\ref{sec:LRtest_design} and its statistical
performance is calculated in Section~\ref{sec:stat_perf_LR_conv}.
Finally, the proposed practical implementation of the Generalized
LRT (GLRT) is presented in Section~\ref{sec:GLRT}. To show the relevance
of the proposed approach, numerical results on large natural image
databases are shown in Section~\ref{sec:numerical_sim}. Section~\ref{sec:conclusion}
concludes the paper.

\section{Detection of LSB Matching Problem Statement.}\label{sec:problem_stat}

This paper mainly focuses on natural images but the extension of the
presented results to any kind of digital media is immediate. Hence,
the column vector $\mC=(c_{1},\ldots,c_{N})^{T}$ represents in this
paper a cover image of $N=N_{x}{\times}N_{y}$ grayscale pixels. The
set of grayscale levels is denoted $\calZ=\{0;\ldots;2^{B-1}\}$ as
pixels values are usually unsigned integers encoded with $B$ bits.
Each cover pixel $c_{n}$ results from the quantization: 
\begin{equation}
	c_{n}=Q(y_{n}),
\label{eq:quant_pixel_value}\end{equation}
where $y_{n}\in\R^{+}$ denotes the raw pixel intensity recorded
by the camera and $Q$ represents the uniform quantization with a
unitary step:
$$
	Q(x)=k\Leftrightarrow x\in[k-1/2\,;\, k+1/2[.
$$
Seeking simplicity, it is assumed in this paper that the saturation
effect is absent, \ie the probability of excessing the quantizer
boundaries $-1/2$ and $2^{B-1}+1/2$ is negligible. Indeed, taking
into account the under or over-exposed pixels is rather simple but
requires a much more complicated notation.\\
 The recorded pixel value can be decomposed as~\cite{Foi08,Cog2011IH}:
\begin{equation}
	y_{n}=\theta_{n}+\xi_{n},
\label{eq:cont_pixel_value}\end{equation}
 where $\theta_{n}$ is a deterministic parameter corresponding to
the mathematical expectation of $y_{n}$ and $\xi_{n}$ is a random
variable representing all the noise corrupting the cover image during
acquisition. As described in~\cite{Foi08}, $\xi_{n}$ is accurately
modeled as a realization of a zero-mean Gaussian random variable $\Xi_{n}\sim\calN(0,\sigma_{n}^{2})$
whose variance $\sigma_{n}^{2}$ varies from pixel to pixel. It thus
follows from~\refeq{eq:quant_pixel_value} and~\refeq{eq:cont_pixel_value}
that $c_{n}$ follows a distribution $P_{\theta_{n}}=P_{\theta_{n},\sigma_{n}}=(p_{\theta_{n}}[0],\ldots,p_{\theta_{n}}[2^{B-1}])$ defined by:
\begin{equation}
	\forall k\in\calZ\;,\; p_{\theta_{n}}[k]=\Phi\left(\frac{k+\nicefrac{1}{2}-\theta_{n}}{\sigma_{n}}\right)-\Phi\left(\frac{k-\nicefrac{1}{2}-\theta_{n}}{\sigma_{n}}\right),
\label{eq:distrib_H0_exact}\end{equation}
with $\Phi$ is the standard Gaussian cumulative distribution function
(cdf) defined by $\Phi(x)=\int_{-\infty}^{x}\phi(u)du$ and $\phi$
the standard Gaussian probability distribution function (pdf) $\phi(u)=\frac{1}{\sqrt{2\pi}}\exp(u^{2}/2)$.
In virtue of the mean value theorem,~\refeq{eq:distrib_H0_exact}
can be written as:
\begin{equation}
	p_{\theta_{n}}[k]=\frac{1}{\sigma_{n}}\int_{k-\frac{1}{2}}^{k+\frac{1}{2}}\phi\left(\frac{u-\theta_{n}}{\sigma_{n}}\right)du=\phi\left(\frac{k-\theta_{n}}{\sigma_{n}}+\epsilon\right),
\label{eq:distrib_H0_approx}\end{equation}
where $\epsilon$ is a (small) corrective term~\cite{Zit2011IH}.

To statistically model stego-image pixels from~\refeq{eq:distrib_H0_exact}--\refeq{eq:distrib_H0_approx},
the two following assumptions are usually adopted~\cite{DMS2004,Fridrich04c}~:
1) the probability of insertion is equal for every cover pixel (independence
between hidden bits and cover pixels) and 2) the message is assumed
compressed and/or cyphered $\mM=(m_{1},\ldots,m_{L})^{T}$ before
insertion. Hence, each hidden bit $m_{l}$ is drawn from a binomial
distribution $\mathcal{B}(1,\nicefrac{1}{2})$, \ie $m_{l}$ is either
$0$ or $1$ with the same probability. This situation is captured
by denoting
\begin{equation}
	\forall n\in\{0,\ldots,N\}\,,\,\left\{
	\begin{array}{c}
		\prob[s_{n}=c_{n}]=(1{-}R),\\
		\prob[s_{n}=c_{n}+\mathrm{ins}(m_{n},c_{n})]=R,
	\end{array}\right.
\label{eq:proba_modif_H0-1}\end{equation}
where $\mS=\{s_{1},\ldots,s_{N}\}$ are the values of stego-image
pixels, the embedding rate $R=\nicefrac{L}{N}$ corresponds to the
number of hidden bits per cover pixel and $\mathrm{ins}(m_{n},c_{n})$
represents the value added to $c_{n}$ to insert the hidden bit $m_{n}$.\\
 The particularity of LSB matching lies in its insertion function
$\mathrm{ins}:\{0;1\}\times\calZ\mapsto\{-1;0;1\}$. Whenever the
LSB of $c_{n}$ is equal to $m_{n}$, \ie when $\mathrm{lsb}(c_{n})=c_{n}\mathrm{mod}2=m_{n}$, there is no need to change $c_{n}$, hence $\mathrm{ins}(m_{n},c_{n})=0$.
On the contrary, whenever $\mathrm{lsb}(c_{n})\neq m_{n}$, the insertion
must change the LSB of $c_{n}$, which is done by adding or subtracting
$1$ with the same probabilities:
\begin{equation}
	\left\{
	\begin{array}{lcl}
		\prob[\mathrm{ins}(b_{s},c_{n})=1\,|\,\mathrm{lsb}(c_{n})\neq m_{n}] & = & \nicefrac{1}{2}\\
		\prob[\mathrm{ins}(b_{s},c_{n})=-1\,|\,\mathrm{lsb}(c_{n})\neq m_{n}] & = & \nicefrac{1}{2}.
	\end{array}
	\right.
\label{eq:proba_insert_H0-1}\end{equation}
 Since each hidden bit $m_{n}$ follows the binomial distribution
$\mathcal{B}(1,\nicefrac{1}{2})$, a straightforward calculation finally
shows that $\prob[\mathrm{lsb}(c_{n})=m_{n}]=\prob[\mathrm{lsb}(c_{n})\neq m_{n}]=\nicefrac{1}{2}$.
Hence, as described in~\cite{Harmsen05pm,CoxDoerr07pm,CoxDoerr08pm,Cog2012SSP},
it follows from~\refeq{eq:proba_modif_H0-1}--\refeq{eq:proba_insert_H0-1}
that for all $n\in\{1,\ldots,N\}$, the pmf of the stego-pixel $s_{n}$
after embedding at rate $R$ with LSB matching is given by $Q_{\theta_{n}}^{R}=\left(q_{\theta_{n}}^{R}[0],\ldots,q_{\theta_{n}}^{R}[2^{b}-1]\right)$
with $\forall k\in\calZ$:
\begin{equation}
	q_{\theta_{n}}^{R}[k]=\frac{R}{4}\left(p_{\theta_{n}}[k{-}1]+p_{\theta_{n}}[k{+}1]\right)+\left(1{-}\frac{R}{2}\right)p_{\theta_{n}}[k].
\label{eq:distrib_H0-1}\end{equation}

\section{Likelihood Ratio Test (LRT) for two simple hypotheses.}\label{sec:LRtest_design}

When analyzing an unknown medium $\mZ$ the first goal of LSB matching
steganalysis is to decide between the two following hypotheses:
\begin{equation}
	\begin{array}{lclcl}
		 &  & \calH_{0} & = & \{z_{n}\sim P_{\theta_{n}}\,,\forall n\in\{1,\ldots,N\}\}\\
		 & \mbox{ vs } & \calH_{1} & = & \{z_{n}\sim Q_{\theta_{n}}^{R}\,,\forall n\in\{1,\ldots,N\}\}.
	\end{array}
\label{eq:match_distr_H1}\end{equation}
 Let us start with the simplest case, when the embedding rate $R$
and, for all $n$, the parameters $\theta_{n}$ and $\sigma_{n}$
are known. In this case, the hypothesis testing problem~\refeq{eq:match_distr_H1}
is reduced to a test between two simple hypotheses.\\
 The goal is obviously to find a test $\delta:\calZ^{N}\mapsto\{\calH_{0},\calH_{1}\}$, such that hypothesis $\calH_{i}$ is accepted if $\delta(\mZ)=\calH_{i}$
(see~\cite{lehman05} for details about statistical hypothesis testing).
However, as explained in the introduction, in an operational forensics
context the most important challenge is first, to warrant a prescribed
(very low) false-alarm probability and second, to maximize the detection
power defined by:
$$
	\beta_{\delta}=\prob_{1}[\delta(\mZ)=\calH_{1}],
$$
where $\prob_{i}(\cdot)$ stands for the probability under hypotheses
$\calH_{i}\,,\, i=\{0;1\}$. Therefore, let $\calK_{\alpha}$ be the
class of tests with an upper-bounded false-alarm probability $\alpha_{0}$
defined by
\begin{equation}
	\calK_{\alpha}=\left\{ \delta:\prob_{0}[\delta(\mZ)=\calH_{1}]\leq\alpha_{0}\right\} .
\label{eq:class_testalpha_def}\end{equation}
In virtue of the Neyman-Pearson lemma, see~\cite[Theorem 3.2.1]{lehman05},
the most powerful (MP) test over the class $\calK_{\alpha_{0}}$~\refeq{eq:class_testalpha_def}
is the LRT given by the following decision rule:
\begin{equation}
	\delta_{R}(\mZ)=
	\left\{
	\begin{array}{lclcl}
		\calH_{0} & \mbox{ if } & \Lambda_{R}(\mZ) & \leq & \tau_{\alpha_{0}}\\
		\calH_{1} & \mbox{ if } & \Lambda_{R}(\mZ) & > & \tau_{\alpha_{0}},
	\end{array}
	\right.
\label{eq:NP_test_simple}\end{equation}
where $\tau_{\alpha_{0}}$ is the solution of $\prob_{0}[\delta(\mZ)>\tau_{\alpha_{0}}]=\alpha_{0}$,
to insure that $\delta_{R}\in\calK_{\alpha_{0}}$, and the likelihood
ratio (LR) $\Lambda_{R}(\mZ)$ is given, from the statistical independence
between pixels, by:
\begin{equation}
	\Lambda_{R}(\mZ)=\prod_{n=1}^{N}\Lambda_{R}(z_{n})=\prod_{n=1}^{N}\frac{R}{4}\frac{p_{\theta_{n}}[z_{n}-1]+p_{\theta_{n}}[z_{n}+1]}{p_{\theta_{n}}[z_{n}]}+\left(1-\frac{R}{2}\right).
\label{eq:match_LR_r}\end{equation}
It can be noted that $\Lambda_{R}(z_{n})$ depends on pixel values
$z_{n}$ through the quantity: 
\begin{equation}
	\Lambda_{2}(z_{n})=\frac{1}{2}\frac{p_{\theta_{n}}[z_{n}-1]+p_{\theta_{n}}[z_{n}+1]}{p_{\theta_{n}}[z_{n}]},
\label{eq:match_LR_1}\end{equation}
 which corresponds to the the likelihood ratio for the conceptual
case of $R=2$. In other words, Equation~\refeq{eq:match_LR_1} corresponds
to this test: $\calH_{0}:\{\, \mZ \mbox{ is a cover medium}\, \}$ vs $\calH_{1}:\{\, \mbox{each pixel of } \mZ \mbox{\, is modified by} \ensuremath{\pm1} \,\}$.
Indeed, considering the case $R\!=\!2$ permits us to clarify the
present methodology, which is then extended to the more general case
of $R\in]0;1[$ in Section~\ref{ssec:stat_perf_LR_anyR}. 

The exact expression for the LR $\Lambda_{2}(z_{n})$ is complicated
due to the corrective terms $\epsilon$ defined in~\refeq{eq:distrib_H0_approx}.
However, the calculation shows that these corrective terms are usually
negligible, particularly when $\sigma_{n}>1$. Therefore, it is proposed
to neglect $\epsilon$ in order to obtain a simplified expression
for the LR $\Lambda_{2}(z_{n})$. From~\refeq{eq:distrib_H0_approx},
this approximation permits us to write:
\begin{eqnarray}
	\frac{p_{\theta_{n}}[z_{n}-1]}{p_{\theta_{n}}[z_{n}]} & = & \exp\left(-\frac{1}{2\sigma_{n}^{2}}\right)\exp\left(\frac{\theta_{n}-z_{n}}{\sigma_{n}^{2}}\right),\nonumber \\
	\frac{p_{\theta_{n}}[z_{n}+1]}{p_{\theta_{n}}[z_{n}]} & = & \exp\left(-\frac{1}{2\sigma_{n}^{2}}\right)\exp\left(\frac{z_{n}-\theta_{n}}{\sigma_{n}^{2}}\right).
\label{eq:match_LR_dvlp}
\end{eqnarray}
Finally, using~\refeq{eq:match_LR_dvlp}, the LR $\Lambda_{2}(z_{n})$
can be written as:
\begin{equation}
	\Lambda_{2}(z_{n})=\frac{1}{4}\exp\left(\frac{-1}{2\sigma_{n}^{2}}\right)\left[\exp\left(\frac{z_{n}-\theta_{n}}{\sigma_{n}^{2}}\right)+\exp\left(\frac{\theta_{n}-z_{n}}{\sigma_{n}^{2}}\right)\right].
\label{eq:match_LR_r_simple}\end{equation}
The logarithm of the likelihood ratio~\refeq{eq:match_log_LR_r_simple}
is usually preferred in order to replace the product in~\refeq{eq:match_LR_r}
with a sum. From~\refeq{eq:match_LR_r_simple}, it immediately follows
that:
\begin{eqnarray}
	\widetilde{\Lambda}_{2}(z_{n}) & \bydef & \log\left[\exp\left(\frac{z_{n}-\theta_{n}}{\sigma_{n}^{2}}\right)+\exp\left(\frac{\theta_{n}-z_{n}}{\sigma_{n}^{2}}\right)\right]
\label{eq:match_log_LR_r_simple}\\
	 & = & \log\big(\Lambda_{2}(z_{n})\big)+\log(2)+\frac{1}{2\sigma_{n}^{2}}.\nonumber
\end{eqnarray}
Again, one can note that the terms $\log(4)$ and $\frac{1}{2\sigma_{n}^{2}}$
do not depend on the true hypothesis. That is why, for the same reasons
as those discussed in connection with Equation~\refeq{eq:match_LR_1},
these terms do not play any role in solving the detection problem~\refeq{eq:match_distr_H1}.
For the sake of clarity, these terms are thus omitted from expression~\refeq{eq:match_log_LR_r_simple}
of the log-LR $\widetilde{\Lambda}_{2}(z_{n})$.

\section{Statistical Performance of the LR test.}\label{sec:stat_perf_LR_conv}
\subsection{Case of simple hypotheses, when $R=2$.}\label{ssec:stat_perf_LR_R1}
In this section it is first proposed to study the statistical performance
for the case of simple hypotheses, when $R=2$. The results are then
extended to the general case of $R\in]0;1[$ in Section~\ref{ssec:stat_perf_LR_anyR}.
To easily calculate the statistical performance of the LR test $\delta_{R}$~\refeq{eq:NP_test_simple},
the asymptotic approach is of crucial interest. Moreover, the assumption
that $N$ grows to infinity is relevant in practice due to the very
large number of pixels in typical images.\\
 For the sake of clarity, let the mean expectation and the mean
variance of $\widetilde{\Lambda}_{2}(z_{n})$ under hypotheses $\calH_{i}$
be defined as follows:
\begin{equation}
	\mu_{i}=\frac{1}{N}\sum_{n=1}^{N}\expec_{i}\Big[\widetilde{\Lambda}_{2}(z_{n})\Big]\;\;\mbox{ and }\;\;\sigma_{i}^{2}=\frac{1}{N}\sum_{n=1}^{N}\var_{i}\Big[\widetilde{\Lambda}_{2}(z_{n})\Big],
\label{eq:mean_var_expect_lambda1}
\end{equation}
where $\expec_{i}\big[\widetilde{\Lambda}_{2}(\mZ)\big]$ and $\var_{i}\big[\widetilde{\Lambda}_{2}(\mZ)\big]$
are respectively the expectation and the variance of $\widetilde{\Lambda}_{2}(z_{n})$
under hypotheses $\calH_{i}\,,\, i=\{0,1\}$. \\
 The test $\widetilde{\delta}_{2}$ associated with the {}``normalized''
log-LR $\widetilde{\Lambda}_{2}(\mZ)$ is defined as:
\begin{equation}
	\widetilde{\delta_{2}}=
	\begin{cases}
		\calH_{0} & \mbox{ if }\;\;\;\widetilde{\Lambda}_{2}(\mZ)\leq\widetilde{\tau}_{\alpha_{0}},\\
		\calH_{1} & \mbox{ if }\;\;\;\widetilde{\Lambda}_{2}(\mZ)>\widetilde{\tau}_{\alpha_{0}}.\end{cases}\;\;\;\mbox{ where }\;\;\;{\displaystyle \widetilde{\Lambda}_{2}(\mZ)\bydef\frac{\sum\limits _{n=1}^{N}\widetilde{\Lambda}_{2}(z_{n})-N\mu_{0}}{\sqrt{N\sigma_{0}^{2}}},}
\label{eq:simplifiedLR_test}\end{equation}
It can noted that the random variables $\widetilde{\Lambda}_{2}(z_{n})$ are assumed statistically independent and, for any $\sigma_{n}>0$, have finite expectation and variance, which implies that the conditions necessary for application of the Lindeberg's central limit theorem~\cite[Theorem 11.2.5]{lehman05} are satisfied.
These conditions can also be shown by using the fact that $z_{n}$
are bounded because they can only take values in the set $\mathcal{Z}$.
Therefore,
\begin{equation}
	\widetilde{\Lambda}_{2}(\mZ)\rightsquigarrow
	\left\{
	\begin{array}{lcl}
		{\displaystyle \calN(0\,,\,1)\;} & \;\mbox{under }\; & \;\calH_{0}\\
		{\displaystyle \calN\left(\frac{\sqrt{N}(\mu_{2}-\mu_{0})}{\sigma_{0}}\,,\,\frac{\sigma_{2}^{2}}{\sigma_{0}^{2}}\right)\;} & \;\mbox{under }\; & \;\calH_{1}.
	\end{array}
	\right.
\label{eq:CLT_logLR}\end{equation}
where $\rightsquigarrow$ represents the convergence in distribution
as $N\to\infty$. From Equation~\refeq{eq:CLT_logLR}, a short algebra
establishes the following theorem. 
\begin{theorem}
\label{thm:thres_simple_test1} For any given probability of false
alarm $\alpha_{0}\in]0;1[$, the decision threshold $\widetilde{\tau}_{\alpha_{0}}$ given by:
\begin{eqnarray}
	\widetilde{\tau}_{\alpha_{0}}=\Phi^{-1}(1-\alpha_{0})
\label{eq:thres_simple_test1}\end{eqnarray}
where $\Phi^{-1}(\cdot)$ is the Gaussian inverse cumulative distribution,
asymptotically warrants that the test $\widetilde{\delta_{2}}$~\refeq{eq:simplifiedLR_test}
is in $\calK_{\alpha_{0}}$. 
\end{theorem}
The main conclusion of Theorem~\ref{thm:thres_simple_test1} is that
the decision threshold $\widetilde{\tau}_{\alpha_{0}}$ depends neither
on the embedding rate $R$ nor the image parameters $\theta_{n}$
and $\sigma_{n}$. Hence, by using the {}``normalized'' log-LR $\widetilde{\Lambda}_{2}(\mZ)$,
the same threshold permits us to respect a prescribed false-alarm
probability $\alpha_{0}$ whatever the analyzed image and the embedding
rate are.\\
\begin{figure}[!t]
\vspace*{-0.3cm}
\centerline{
\hspace*{-0.75cm}
  \subfloat[Expectations of the log-LR $\log\big( \Lambda_2(z_n) \big) $ as a function of expectation $\theta$.]{
  \hspace*{-0.15cm}
  \def\svgwidth{0.45\textwidth}
    \ifpdf
      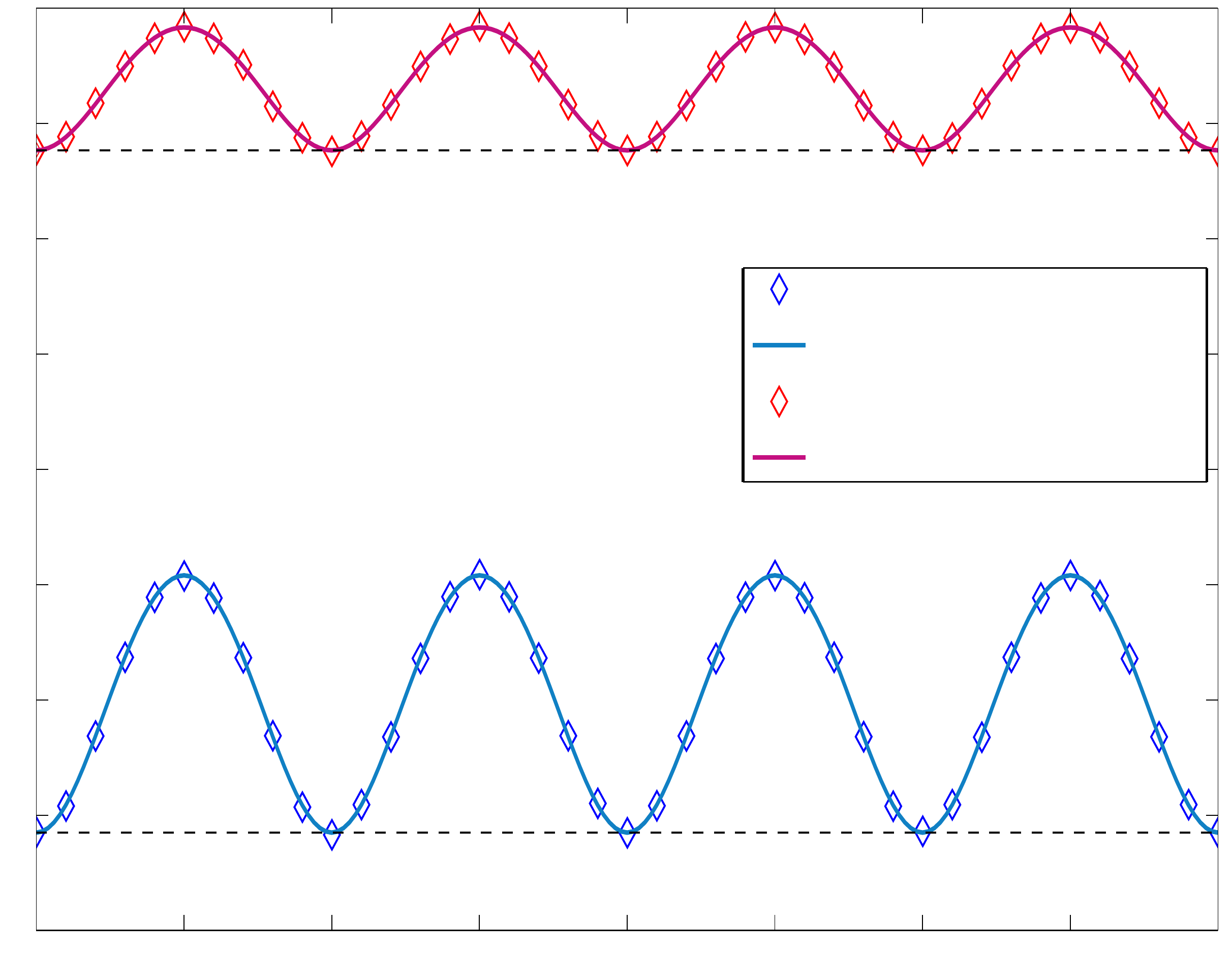
    \else
    \fi
  \label{fig:detect_theo_match_mean_quantifs}  }
\hspace*{0.35cm}
  \subfloat[Variances of the log-LR $\log\big( \Lambda_2(z_n) \big) $ as a function of expectation $\theta$.]{
  \def\svgwidth{0.45\textwidth}
  \hspace*{-0.15cm}
    \ifpdf
      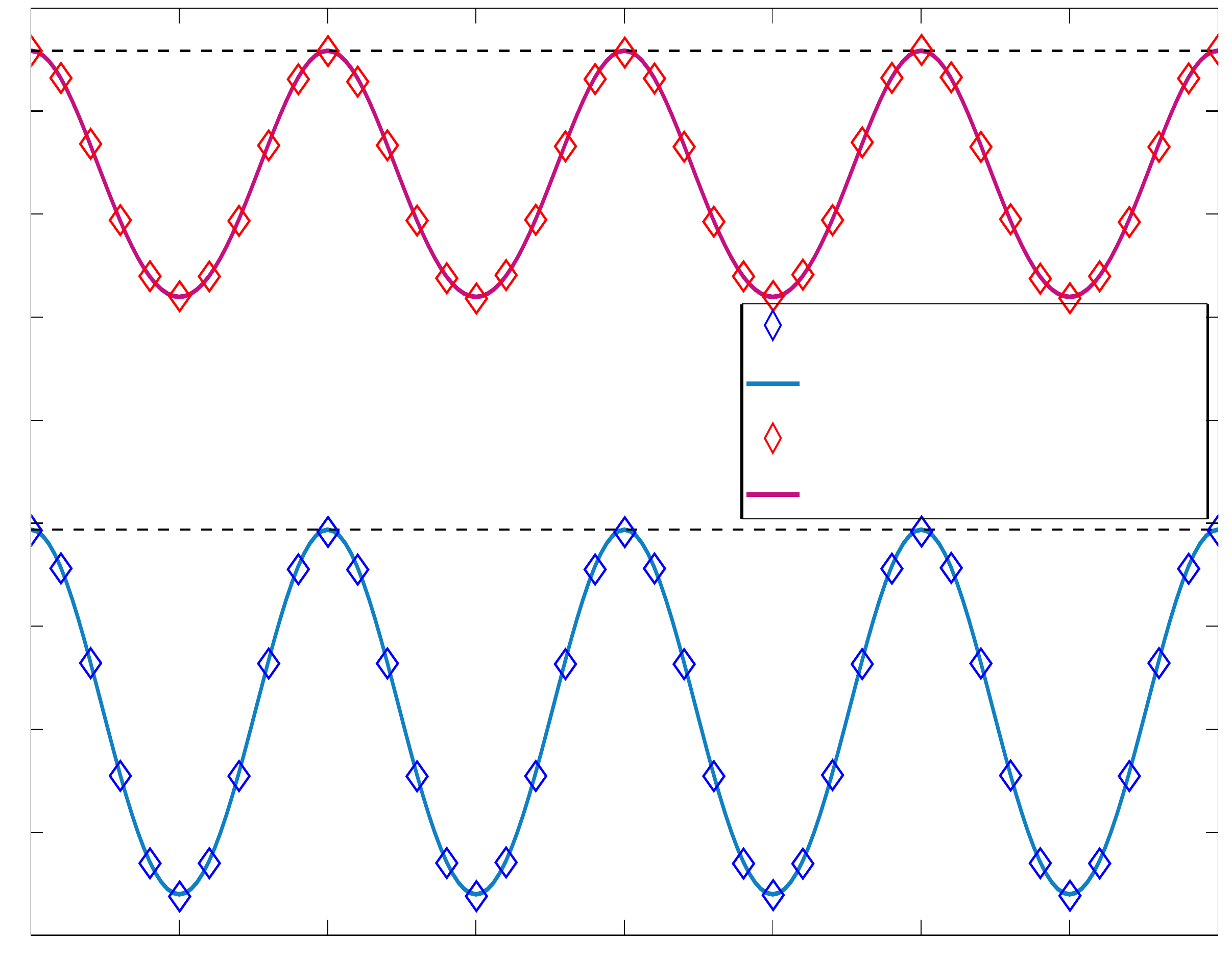
    \else
    \fi
  \label{fig:detect_theo_match_var_quantifs}  }
  }
\caption{Graphical representation of the two first moments of log-LR $\log\big( \Lambda_2(z_n) \big)$~\refeq{eq:match_H0_mean0} - \refeq{eq:match_H1_var1}. Presented results  correspond to the case of i.i.d pixels with expectation $\theta_n\in[126; 130]$ and standard deviation $\sigma_n=0.75$. }
\label{fig:detect_theo_match_moments}
\end{figure}
Equation~\refeq{eq:CLT_logLR} also implies that to asymptotically
calculate the detection power of LR test $\widetilde{\delta_{2}}$~\refeq{eq:simplifiedLR_test}, one only needs to calculate the first moments of $\widetilde{\Lambda}_{2}(\mZ)$.
The mean expectations used in the log-LR $\widetilde{\Lambda}_{2}(z_{n})$
are given under hypotheses $\calH_{0}$ and $\calH_{1}$ by
\begin{eqnarray}
	\mu_{0} & = & \frac{1}{N}\sum_{n=1}^{N}\sum_{k\in\calZ}p_{\theta_{n}}[k]\log\left(\exp\left(\frac{k-\theta_{n}}{\sigma_{n}^{2}}\right)+\exp\left(\frac{\theta_{n}-k}{\sigma_{n}^{2}}\right)\right),\;\;\;\;\;\;\;
\label{eq:match_H0_mean0}\\
	\mu_{2} & = & \frac{1}{N}\sum_{n=1}^{N}\sum_{k\in\calZ}q_{\theta_{n}}^{R}[k]\log\left(\exp\left(\frac{k-\theta_{n}}{\sigma_{n}^{2}}\right)+\exp\left(\frac{\theta_{n}-k}{\sigma_{n}^{2}}\right)\right),\;\;\;\;\;\;\;
\label{eq:match_H1_mean1}\end{eqnarray}
 where the probabilities $p_{\theta_{n}}[k]$ and $q_{\theta_{n}}^{R}[k]$
are respectively defined in~\refeq{eq:distrib_H0_exact} and~\refeq{eq:distrib_H0-1}.\\
Similarly, the mean variances are by definition given under both
hypotheses $\calH_{0}$ and $\calH_{1}$ by:
\begin{eqnarray}
	\sigma_{0}^{2} & \,=\, & \frac{1}{N}\sum_{n=1}^{N}\sum_{k\in\calZ}p_{\theta_{n}}[k]\log\left(\exp\left(\frac{k{-}\theta_{n}}{\sigma_{n}^{2}}\right)+\exp\left(\frac{\theta_{n}{-}k}{\sigma_{n}^{2}}\right)\right)^{2}-\mu_{0}^{2},\;\;\;\;\;\;\;\;
\label{eq:match_H0_var0}\\
	\sigma_{2}^{2} & \,=\, & \frac{1}{N}\sum_{n=1}^{N}\sum_{k\in\calZ}q_{\theta_{n}}^{R}[k]\log\left(\exp\left(\frac{k{-}\theta_{n}}{\sigma_{n}^{2}}\right)+\exp\left(\frac{\theta_{n}{-}k}{\sigma_{n}^{2}}\right)\right)^{2}-\mu_{2}^{2}.\;\;\;\;\;\;\;\;
\label{eq:match_H1_var1}\end{eqnarray}
The expectations $\mu_{0}$ and $\mu_{2}$ and the variances $\sigma_{0}^{2}$
and $\sigma_{2}^{2}$ as functions of $\theta_{n}$ are respectively
drawn in Figures~\ref{fig:detect_theo_match_mean_quantifs} and~\ref{fig:detect_theo_match_var_quantifs}.
These figures highlight the fact that the pixel expectation $\theta_{n}$
can have a significant impact on the LR moments, and later on the
detection power, particularly when $\sigma_{n}<1$. However, a thorough
study of equations~\refeq{eq:match_H0_mean0}--\refeq{eq:match_H1_var1}
shows that this phenomenon rapidly tends to be negligible when $\sigma_{n}\gtrsim1$.\\
Even thoug, the moments given in~\refeq{eq:match_H0_mean0}--\refeq{eq:match_H1_var1}
have a rather complicated expression, their numerical calculation
is straightforward as long as the parameters $\theta_{n}$ and $\sigma_{n}$
are known.

From the asymptotic distribution~\refeq{eq:CLT_logLR} of the log-LR
$\widetilde{\Lambda}_{2}(\mZ)$ and the expressions~\refeq{eq:match_H0_mean0}--\refeq{eq:match_H1_var1}
of its two first moments, the detection power of the LR test $\widetilde{\delta_{2}}$~\refeq{eq:simplifiedLR_test}
is given by the following theorem.
\begin{theorem}\label{thm:LR_simple_power1} 
For any $\alpha_{0}\in]0;1[$, assuming that the parameters $\{\theta_{n}\}_{n=1}^{N}$ and $\{\sigma_{n}\}_{n=1}^{N}$ are known, the power function $\widetilde{\beta}_{\delta_{2}}$ associated with the test $\widetilde{\delta}_{2}$~\refeq{eq:simplifiedLR_test} is asymptotically given, as $N\to\infty$, by:
\begin{equation}
	\widetilde{\beta}_{\delta_{2}}=1-\Phi\left(\frac{\sigma_{0}}{\sigma_{2}}\Phi^{-1}(1-\alpha_{0})+\frac{\sqrt{N}(\mu_{0}-\mu_{2})}{\sigma_{2}}\right).
\label{eq:LR_simple_power1}\end{equation}
\end{theorem}

\begin{proof}
Using the result~\refeq{eq:CLT_logLR}, it asymptotically holds that
for any $\widetilde{\tau}_{\alpha_{0}}\in\R$: 
$$
	\alpha_{0}(\widetilde{\delta}_{2})=\prob_{0}\left[\widetilde{\Lambda}_{2}(\mZ)>\widetilde{\tau}_{\alpha_{0}}\right]=1-\Phi\left(\widetilde{\tau}_{\alpha_{0}}\right).
$$
Hence, because $\Phi$ is strictly increasing, one has:
\begin{equation}
	(1-\alpha_{0}(\widetilde{\delta}_{2}))=\Phi(\widetilde{\tau}_{\alpha_{0}})\Leftrightarrow\widetilde{\tau}_{\alpha_{0}}=\Phi^{-1}\left(1-\alpha_{0}(\delta_{2})\right),
\end{equation}
which proves Theorem~\ref{thm:thres_simple_test1}.\\
It also follows from~\refeq{eq:CLT_logLR} that for any decision
threshold $\widetilde{\tau}_{\alpha_{0}}\in\R$ the power of the test
$\widetilde{\delta_{2}}$~\refeq{eq:simplifiedLR_test} is given by:
$$
	\widetilde{\beta}_{\delta_{2}}=\prob_{1}\left[\widetilde{\Lambda}_{2}(\mZ)>\widetilde{\tau}_{\alpha_{0}}\right]=1-\Phi\left(\frac{\sigma_{0}}{\sigma_{2}}\bigg(\widetilde{\tau}_{\alpha_{0}}-\frac{\sqrt{N}(\mu_{2}-\mu_{0})}{\sigma_{0}}\bigg)\right).
$$
By substituting $\widetilde{\tau}_{\alpha_{0}}$ by the value given
in Theorem~\ref{thm:thres_simple_test1}, a short algebra leads to
the relation~\refeq{eq:LR_simple_power1}. This proves Theorem~\ref{thm:LR_simple_power1}
and concludes the proof.
\end{proof}

\subsection{General case of $R\in]0;1[$.}\label{ssec:stat_perf_LR_anyR}
The case for which the embedding rate $R$ can take any value in $]0;1]$
is treated in a similar manner as the case $R=2$. The problem of
designing an optimal test has been shown to be particularly difficult
in~\cite{Zit2011IH}. A thorough design a MP test uniformly with
respect to the embedding rate lies outside of the scope of this paper
which mainly studies the MP test for $R=2$ and its practical implementation.
Hence, it is proposed to use the test $\widetilde{\delta_{2}}$~\refeq{eq:simplifiedLR_test}
whatever the embedding rate $R$ might be. Once again, the asymptotic
distribution~\refeq{eq:CLT_logLR} is used to solve the decision
problem~\refeq{eq:match_distr_H1}.

The alternative hypothesis $\calH_{R}$, that $\mZ$ contains a stego-medium
with embedding rate $R\in]0;1]$, can be considered as a combination
of stego and cover pixels. Hence, the use of the law of total expectation
and the law of total variance is relevant to calculate the two first
moments of the log-LR $\widetilde{\Lambda}_{2}(\mZ)$. Using the moments
given in~\refeq{eq:match_H0_mean0}--\refeq{eq:match_H1_var1}, for
the case $R=2$, a short calculation gives:
\begin{eqnarray}
	\mu_{R} & = & \frac{R}{2}\mu_{2}+\left(1-\frac{R}{2}\right)\mu_{0},\label{eq:match_HR_moments}\\
	\sigma_{R}^{2} & = & \frac{R}{2}(\sigma_{2}^{2}+\mu_{2}^{2})+\left(1-\frac{R}{2}\right)(\sigma_{0}^{2}+\mu_{0}^{2})-\left(\frac{R}{2}\mu_{2}+\Big(1-\frac{R}{2}\Big)\mu_{0}\right)^{2}.
\end{eqnarray}

\begin{figure}[!t]
\vspace*{-0.3cm}
\centerline{
\hspace*{-0.75cm}
  \subfloat[Detection power as a function of false alarm probability $\alpha_0$ (ROC curves).]{
  \hspace*{-0.15cm}
  \def\svgwidth{0.45\textwidth}
    \ifpdf
      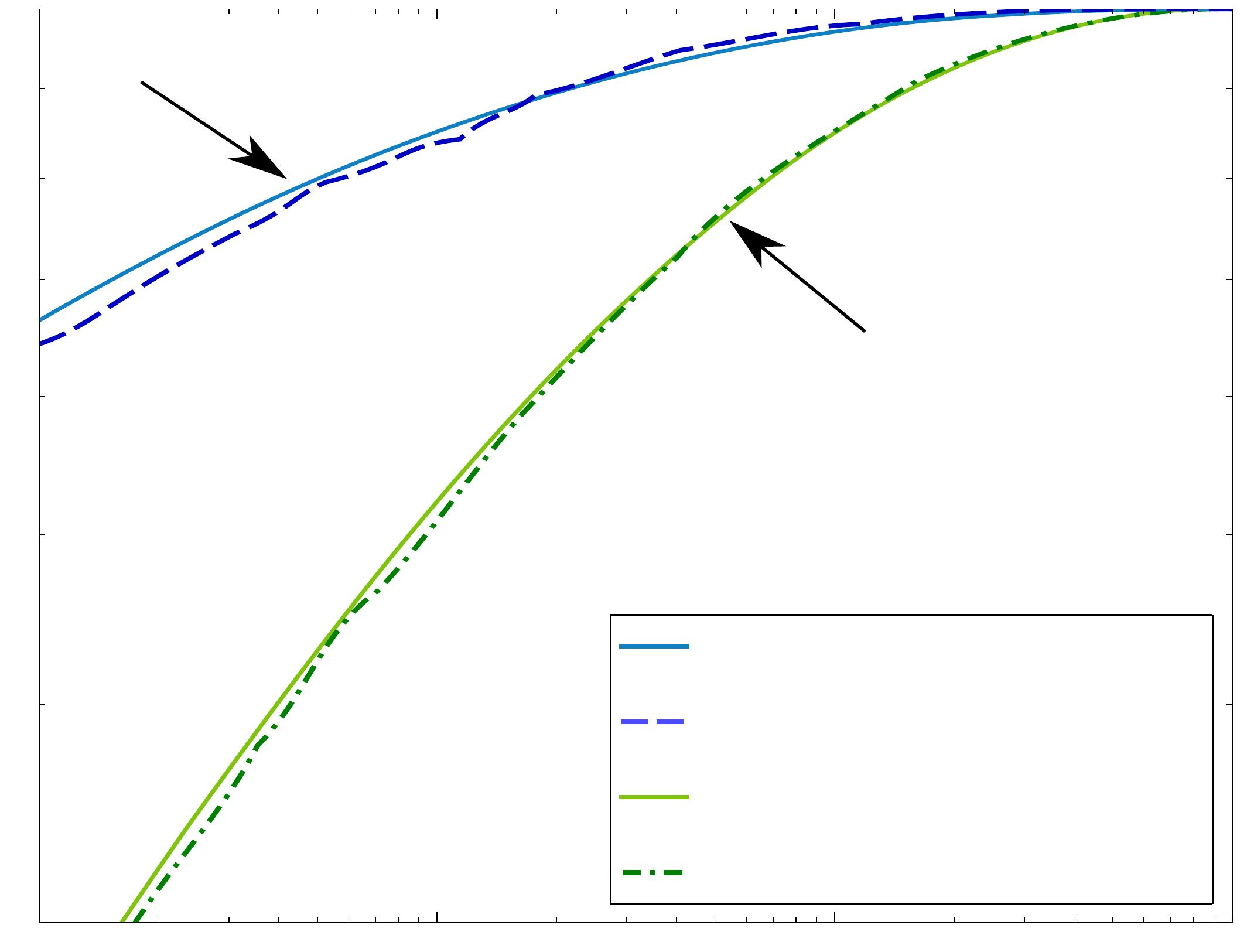
    \else
    \fi
  \label{fig:detect_theo_match_simul_beta_vs_alpha}  }
\hspace*{0.35cm}
  \subfloat[False alarm probability $\alpha_0$ as a function of threshold value $\tau_{\alpha_0}$.]{
  \def\svgwidth{0.45\textwidth}
  \hspace*{-0.15cm}
    \ifpdf
      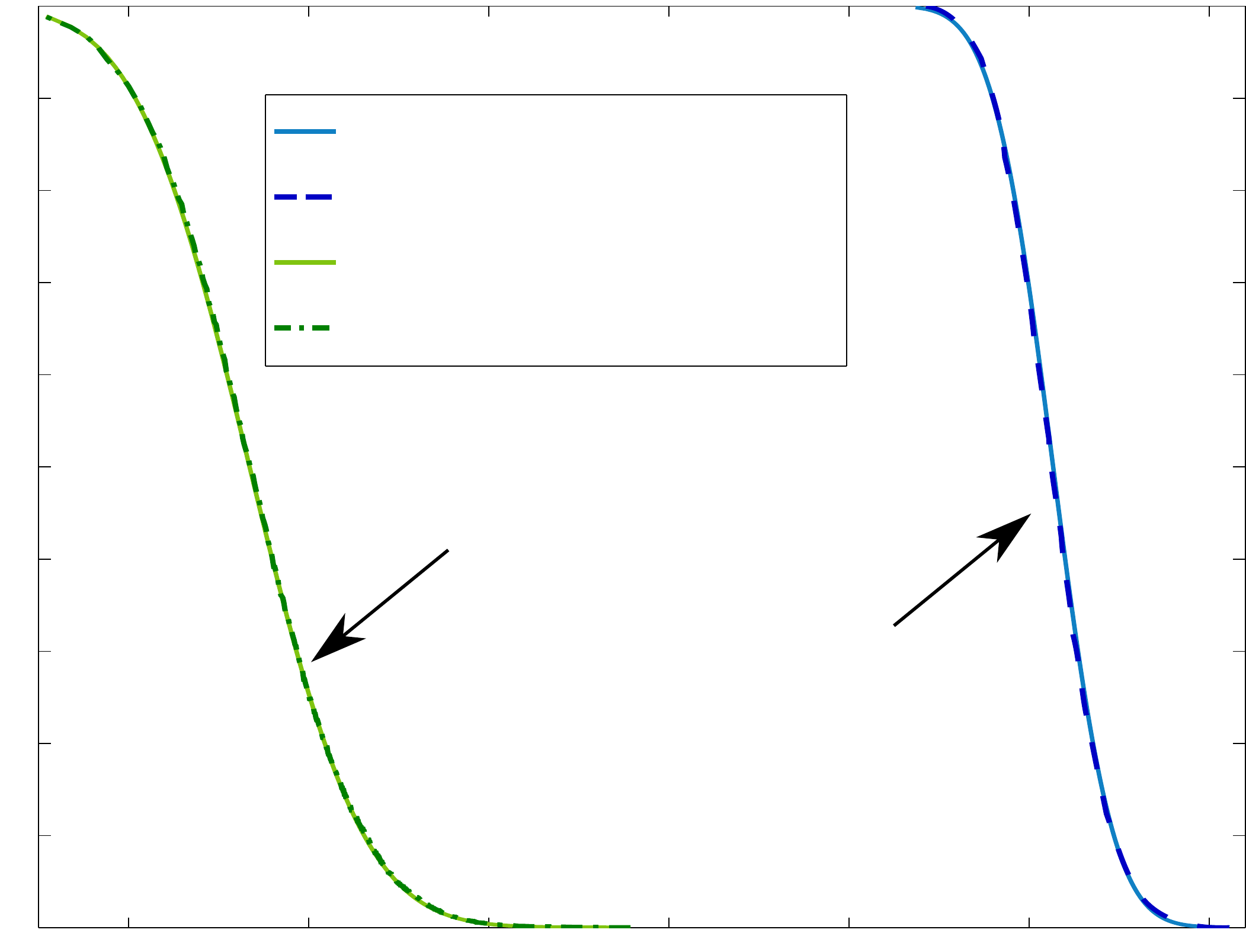
    \else
    \fi
  \label{fig:detect_theo_match_simul_alpha_vs_seuil}  }
  }
\caption{Illustration of LRT statistical performance, false-alarm probabilities and detection power, for $N=1000$ pixels, $R=0.1$, $\sigma_n=0.5$ and $\theta=\{127.5 ; 128\}$. The empirical results were obtained with $5.10^4$ realizations.}
\label{fig:detect_theo_match_performances}
\end{figure}

In other words, by using the test $\widetilde{\delta_{2}}$~\refeq{eq:simplifiedLR_test}
for any $R\in]0;1]$ only the detection power is impacted. Indeed,
the null hypothesis does not change, hence, the asymptotic distribution~\refeq{eq:CLT_logLR} of the LR $\widetilde{\Lambda}_{2}(\mZ)$ under $\calH_{0}$ as well
as the decision threshold $\widehat{\tau}_{\alpha_{0}}$~\refeq{eq:thres_simple_test1} remain the same. This point is highlighted in the following theorem. 
\begin{theorem}
\label{thm:LR_simple_powerR} For any $\alpha_{0}\in]0;1[$, assuming
that the parameters $\{\theta_{n}\}_{n=1}^{N}$ and $\{\sigma_{n}\}_{n=1}^{N}$
are known, the power function $\widetilde{\beta}_{\delta_{R}}$ associated
with the test $\widetilde{\delta}_{2}$~\refeq{eq:simplifiedLR_test}
is asymptotically given for any $R\in]0;1]$ by:
\begin{equation}
	\widetilde{\beta}_{\delta_{R}}=1-\Phi\left(\frac{\sigma_{0}}{\sigma_{R}}\Phi^{-1}(1-\alpha_{0})+\frac{R\sqrt{N}(\mu_{0}-\mu_{2})}{\sigma_{R}}\right).
\label{eq:power_simple_testR}\end{equation}
\end{theorem}
The power functions $\widetilde{\beta}_{\delta_{R}}$ for $N\!=\!1000$,
$R\!=\!0.1$, $\sigma_{n}\!=\!0.5$ and $\theta_{n}\!=\!\{127.5;128\}$
are drawn in Figure~\ref{fig:detect_theo_match_simul_beta_vs_alpha}.
Once again, this figure highlights the potentially significant impact
of pixel expectation on the performance of the test $\widetilde{\delta}_{2}$.\\
 It should be highlighted that the most powerful property of the
test $\widetilde{\delta}_{2}$ is difficult to prove for $R\in]0;1[$,
see~\cite{Cog2011ISIT}. However, Figure~\ref{fig:detect_theo_clairvoyant}
emphasizes the relevance of the proposed approach, which consists
in designing a test for $R=2$ and extending its application to $R\in]0;1[$.
Here, the power function of the proposed test is compared with the
power function of the clairvoyant detector, that knows $R$. The numerical
comparison present in Figure~\ref{fig:detect_theo_clairvoyant} shows
that the loss of the power is negligible.

Finally, it can be noted that the detection power as given in Theorem~\ref{thm:LR_simple_powerR}
complies with the square root law of steganographic capacity~\cite{ker07}.
Indeed, from~\refeq{eq:power_simple_testR}, a short algebra immediately
permits us to establish that:
\begin{equation}
	\lim\limits _{\sqrt{N}/L\to0}\widetilde{\beta}_{\delta_{R}}=1\;\;\mbox{ and }\;\;\lim\limits _{\sqrt{N}/L\to\infty}\widetilde{\beta}_{\delta_{R}}=\alpha_{0}.
\label{eq:_LR_simple_powerR_SRL}
\end{equation}

\begin{figure}[!t]
\vspace*{-0.3cm}
\centerline{
  \def\svgwidth{0.5\textwidth}
  \hspace*{-0.15cm}
    \ifpdf
      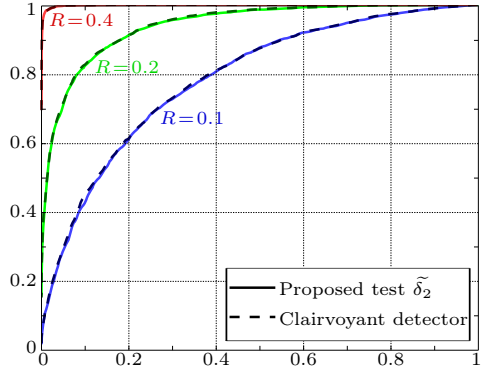
    \else
    \fi
  }
\caption{Numerical comparison between Proposed LR test $\widetilde{\delta_2}$~\refeq{eq:simplifiedLR_test}, and the clairvoyant detector which knows the embedding rate $R=0.1$ ans, thus, uses the LR test design for this rate. Results were obtained from a Monte-Carlo simulation with $5.10^4$ realizations using \emph{Lena} image cropped to $128\times128$ pixels and addition of a Gaussian white noise with $\sigma=2$.}
\label{fig:detect_theo_clairvoyant}
\end{figure}

\section{Practical implementation of proposed LR test.}\label{sec:GLRT}

In a practice, the application of the test $\widetilde{\delta_{2}}$~\refeq{eq:simplifiedLR_test}
is compromised because neither the expectation $\theta_{n}$ nor the
variance $\sigma_{n}^{2}$ of pixels are known: their estimated values,
denoted $\widehat{\theta}_{n}$ and $\widehat{\sigma}_{n}^{2}$, respectively,
have to be used instead.

However, accurate estimation of the parameters $\theta_{n}$ and $\sigma_{n}$
is a difficult problem but necessary to obtain a high detection performance.
This problem also occurs in LSB replacement steganalysis. An efficient
yet simple way to overcome this problem was introduced in the well-known
Weighted Stego-image steganalysis (WS), initially proposed in~\cite{Fridrich04c}.
The authors propose to locally estimate the parameter $\theta_{n}$
by filtering the inspected image so that $\widehat{\theta}_{n}$ correspond
to the mean of the four surrounding pixels. Similarly, the local variance
of the four surrounding pixels is used to estimate $\sigma_{n}^{2}$.
The WS method has been studied thoroughly in~\cite{Ker08a} and two
major improvements have been proposed. First, the authors have empirically
enhanced the estimation of pixel expectations by testing different
local filters. Second, the author proposed to use moderated weights
$w_{n}=\widehat{\sigma}_{n}^{2}+\alpha\,,\,\alpha>0$ instead of the
variance estimation $\widehat{\sigma}_{n}^{2}$.

In the present paper, it is proposed to use the WS filtering method
to estimate the parameters $\theta_{n}$ and $\sigma_{n}^{2}$. Note
that the proposed practical test is not optimal but intends to show
the relevance of the proposed approach and feasibility to design a
practical efficient test. Following the WS method, the practical implementation
of the LR test $\widehat{\delta}_{2}$ proposed in this paper estimates
each $\theta_{n}$ by filtering the inspected image with the kernel:
$$\frac{1}{4}
  \begin{pmatrix}
    -1 & 2 & -1 \\
     2 & 0 & 2  \\
    -1 & 2 & -1
  \end{pmatrix}
$$
 Contrary to what is suggested in~\cite{Ker08a}, for the case of
LSB replacement, our numerical experiments indicate that the detection
performance tends to get worse when using the moderated weights instead
of the estimated variance. Our interpretation of this phenomenon is
as follows. The proposed LR test~\refeq{eq:simplifiedLR_test} essentially
relies on the increase of pixels' variance due to insertion of hidden
information. Hence, the use of moderated weights tends to fundamentally
bias the test and deflates the performance results. Figure~\ref{fig:BOSS_COR_comp_alpha}
offers an example of this phenomenon through a comparison of ROC curves
obtained using 10\,000 images from the BOSSbase database with $R=1/2$
and $\alpha=\{1/4;1/2;3/4;1\}$.\\
 On the other hand, the direct use of the estimated variance $\widetilde{\sigma}_{n}^{2}$
may lead to numerical instability particularly in flat image areas.
Hence, it was chosen to add $\alpha=1/4$ to the estimated variance
in our numerical experiments.

By using these estimated values in expression~\refeq{eq:match_log_LR_r_simple}
the estimated log-LR $\widetilde{\Lambda}_{2}(z_{n})$, see Equation~\refeq{eq:match_log_LR_r_simple}
becomes:
\begin{equation}
	\widehat{\Lambda}_{2}(z_{n})=\log\left[\exp\left(\frac{z_{n}-\widehat{\theta}_{n}}{(\alpha+\widehat{\sigma}_{n})^{2}}\right)+\exp\left(\frac{\widehat{\theta}_{n}-z_{n}}{(\alpha+\widehat{\sigma}_{n})^{2}}\right)\right].
\label{eq:match_log_GLR_r_simple}\end{equation}
It should be highlighted that some difficult problems still remain
open.\\
First, the normalization of the log-LR, suggested in Equation~\refeq{eq:simplifiedLR_test},
requires the calculation of the expectation $\mu_{0}$ and the variance
$\sigma_{0}^{2}$ of the log-LR. Unfortunately, the estimates of the
parameters $\sigma_{n}$ are, in practice, not accurate enough to
perform this normalization efficiently.\\
Second, possibly the most difficult problem is that the statistical
inference between the cover image and the hidden information should
be taken into account. For instance it was proposed in~\cite{Zit2011IH}
to remove the LSB plane in order to remove any potential stego-noise.
For LSB matching this is not possible. Therefore, the impact of hidden
information on estimators $\widehat{\theta}_{n}$ and $\widehat{\sigma}_{n}$
should be studied. Since the proposed test relies mainly on the slight
increase of pixels' variance due to data hiding, the embedding changes
may have an important effect on the estimates $\widehat{\sigma}_{n}$
and on the proposed test.

\begin{figure}[!t]
\centerline{
\hspace*{-0.5cm}
  \subfloat[ROC obtained with four different weight factor: $\alpha=\{1/4 \,;\, 1/2 \,;\, 3/4 \,;\, 1 \}$.]{
  \hspace*{-0.55cm}
  \def\svgwidth{0.45\textwidth}
    \ifpdf
      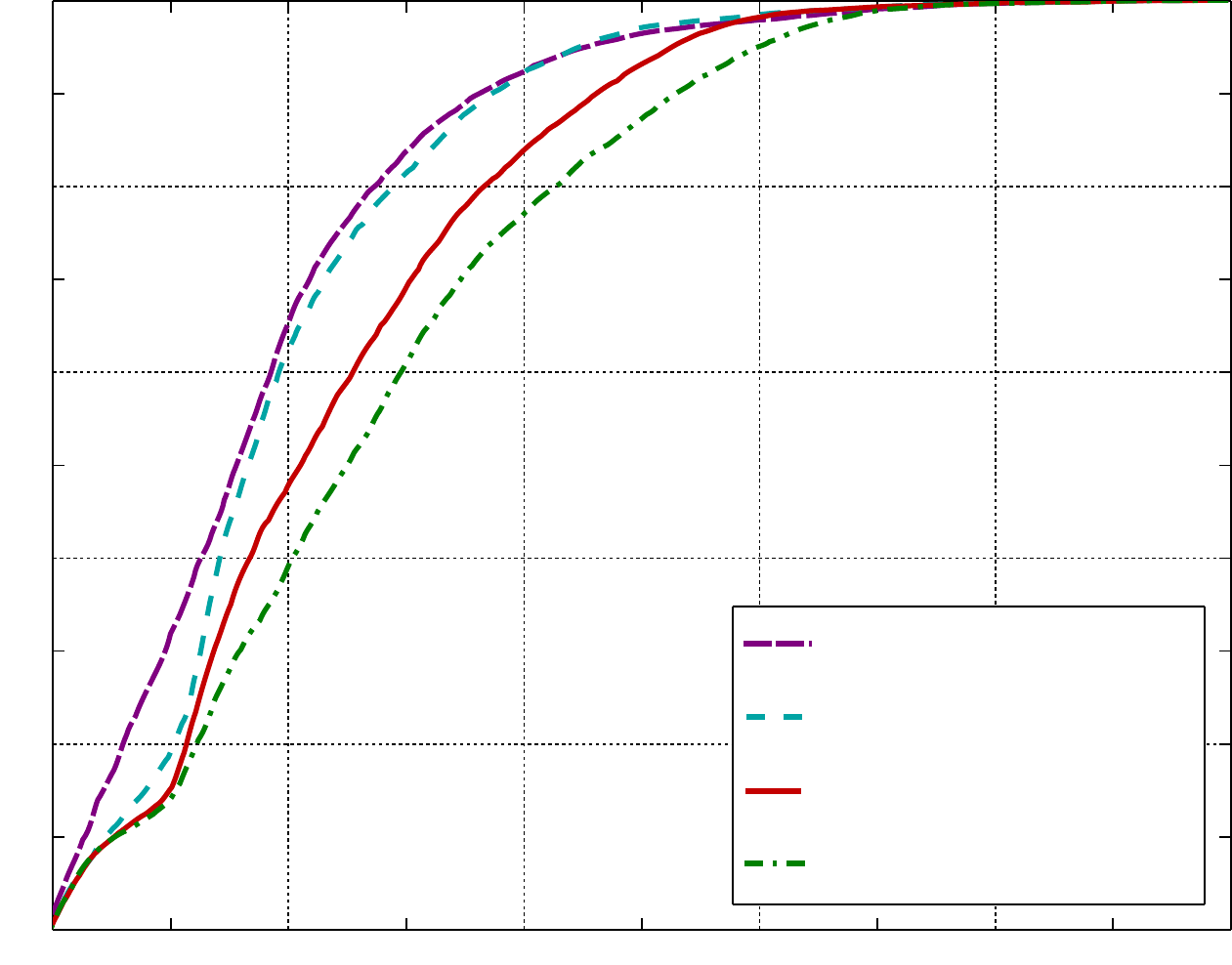
    \else
    \fi
  \label{fig:BOSS_COR_comp_alpha}  }
\hspace*{0.35cm}
  \subfloat[ROC curves obtained with the two statistics $\widehat{\Lambda}_{2}$ and $\widehat{\Lambda}_{2}^{\star}$.]{
  \def\svgwidth{0.45\textwidth}
  \hspace*{-0.55cm}
    \ifpdf
      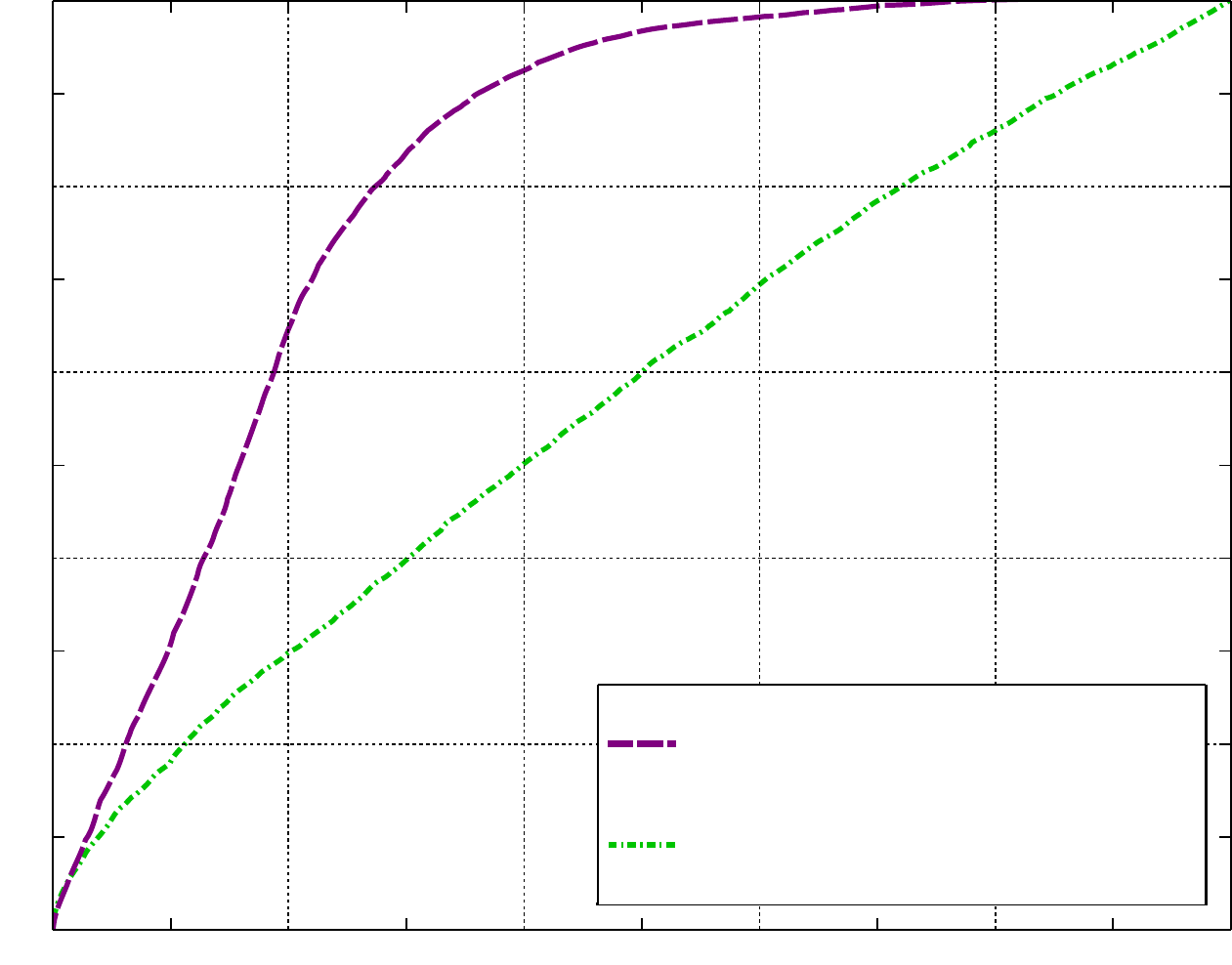
    \else
    \fi
  \label{fig:BOSS_COR_comp_LR}  }
  }
\caption{Impact of weights and calibration on proposed test performance. ROC curves obtained using the images from BOSS database~\cite{BOSS2010} with $R=0.5$.}
\label{fig:BOSS_COR_comp}
\end{figure} 

As explained above, proper normalization of the proposed test is critical
in practice. Even though the proposed LR is very sensitive to hidden
information, if its expectation can not be set to a fixed value under
$\calH_{0}$, the between-image-error described in~\cite{BohmeKer2006}
may negatively impact the test accuracy. Numerical simulations show
that the expectation of the LR $\widehat{\Lambda}_{2}(z_{n})$ can
be roughly approximated by $-\log(2)-\frac{1}{4\widehat{\sigma}_{n}^{2}}$.\\
 Therefore, the practical test proposed in the present paper is
given as:
\begin{eqnarray}
	& & \widehat{\delta}_{2}=
	\begin{cases}
		\calH_{0} & {\displaystyle \mbox{ if }\;\;\;\widehat{\Lambda}_{2}^{\star}(\mZ)\leq\widehat{\tau}_{\alpha_{0}},}\\
		\calH_{1} & {\displaystyle \mbox{ if }\;\;\;\widehat{\Lambda}_{2}^{\star}(\mZ)>\widehat{\tau}_{\alpha_{0}},}
	\end{cases}
\label{eq:simplifiedGLR_test}\\
	&\mbox{ with }& \widehat{\Lambda}_{2}^{\star}(\mZ)=\frac{1}{\sqrt{N}}\sum_{n=1}^{N}\widehat{\Lambda}_{2}(z_{n})-\log(2)-\frac{1}{4(\alpha+\widehat{\sigma}_{n})^{2}}.
\end{eqnarray}
 One can note that, contrary to the LR statistically studied throughout
Sections~\ref{ssec:stat_perf_LR_R1}--\ref{ssec:stat_perf_LR_anyR},
the proposed decision statistic is not normalized. Indeed the variance
of $\widehat{\Lambda}_2(z_{n})$ is not taken into account in Equation~\ref{eq:simplifiedGLR_test}.
This is because the estimation of pixels' variance is particularly
difficult and the method used in this paper is not accurate enough.
In fact, normalization can even lower the detection performance. The
most notable thing about the test~\refeq{eq:simplifiedGLR_test}
is that the expectation of the decision statistics $\widehat{\Lambda}_2^{\star}(\mZ)$
is always $0$ under hypothesis $\calH_{0}$. Figure~\ref{fig:BOSS_COR_comp_LR}
shows an example of the detection power obtained with the two tests
based on the statistics~\refeq{eq:match_log_GLR_r_simple} and~\refeq{eq:simplifiedGLR_test}.

\section{Numerical Simulations.}\label{sec:numerical_sim}
\subsection{Theoretical results on simulated data.}

\begin{figure}[!t]
\vspace*{-0.85cm}
\centerline{
\hspace*{-0.75cm}
      \subfloat[Digital image used for the Monte-Carlo simulations]{
	\hspace*{-0.15cm}
	\ifpdf
	  \includegraphics[width=0.41\textwidth, height=4.25cm]{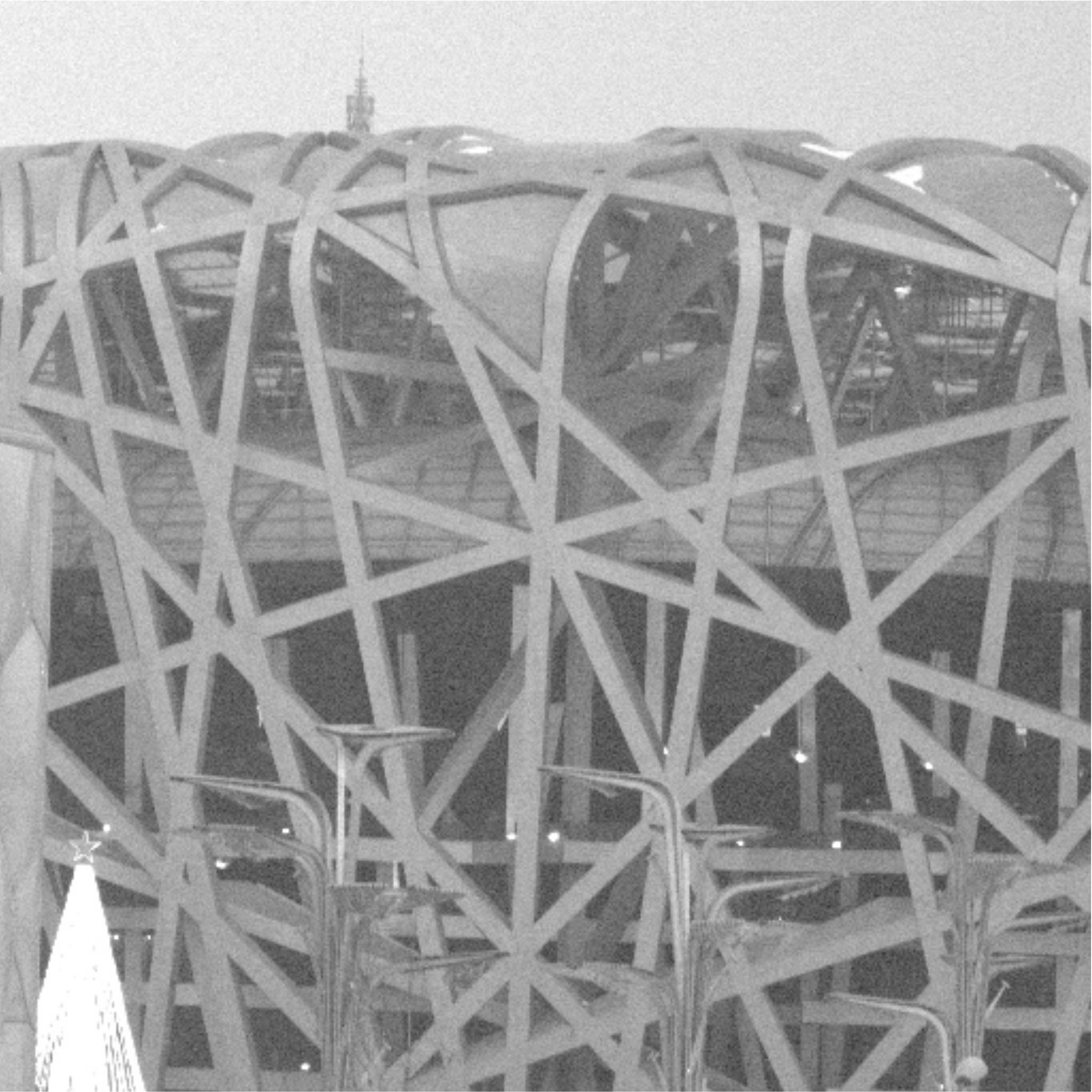}
	\else
	\fi
	\label{fig:MC_img_pwr_NO}}
\hspace*{0.35cm}
      \subfloat[Power of the test $\widehat{\delta}_2$~\refeq{eq:simplifiedGLR_test} as a function of pixel number for different false-alarm probabilities: theory and simulation.]{
	\def\svgwidth{0.45\textwidth}
	\ifpdf
	  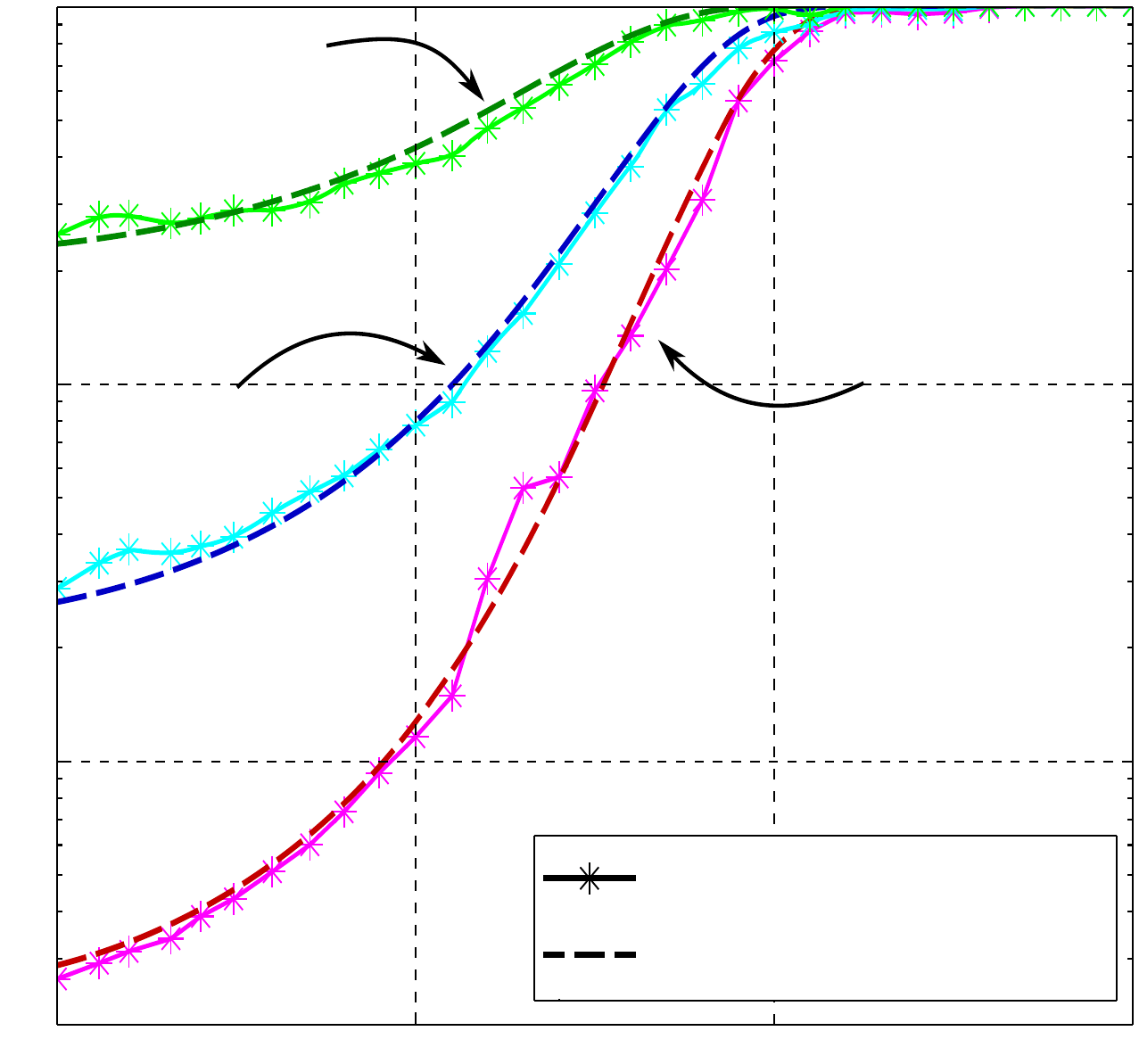
	\else
	\fi
	\label{fig:MC_img_pwr_vs_alpha_ROC}}
  }
\caption{Numerical verification of theoretical results through Monte-Carlo simulation based on natural image shown in Figure~\ref{fig:MC_img_pwr_NO}.}
\label{fig:MC_img_pwr_all}
\end{figure}

Figure~\ref{fig:MC_img_pwr_all} presents a numerical verification
of Theorem~\ref{thm:LR_simple_powerR}. The image shown in Figure~\ref{fig:MC_img_pwr_NO} has been analyzed $5.10^{4}$ times. Each run was preceded by the addition of a zero-mean Gaussian noise whose standard deviation was
$\sigma=1$. The embedded hidden information was drawn from a binomial
distribution $\mathcal{B}(1,1/2)$ with an embedding rate $R=1$.
The empirical power of the test $\widehat{\delta}_{2}$ is compared
with the theoretical result given by Theorem~\ref{thm:LR_simple_powerR}
for three different false-alarm probabilities: $\alpha_{0}=\{10^{-1};10^{-2};10^{-3}\}$. Observe that the obtained detection power almost perfectly corresponds
to the theoretical results.\\
 Note that it is crucial to use the same image for this Monte-Carlo
simulation because the detection power of the proposed test depends
on image parameters, namely on $\theta_{n}$ and particularly on $\sigma_{n}^{2}$.
Hence, for a different image, the detection power may differ significantly
as explained in Section~\ref{sec:stat_perf_LR_conv}. Moreover, the
use of the same image artificially permits us to overcome the difficult
problem of normalizing the log-LR and, thus, the effects of the between-image-error
described in~\cite{BohmeKer2006}.

\subsection{Comparison with the state of the art on real images.}
Matlab source code of proposed test, as detailed in Equation~\refeq{eq:simplifiedGLR_test}, is available on the Internet at~:
\url{http://remi.cogranne.pagesperso-orange.fr/}.
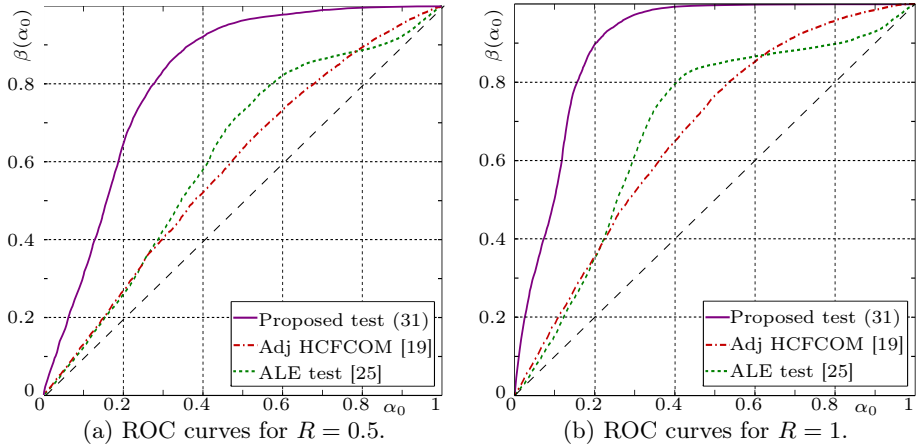
\begin{figure}[!t]
\centerline{
\hspace*{-0.5cm}
  \subfloat[ROC curves for $R=0.5$.]{
  \hspace*{-0.55cm}
  \def\svgwidth{0.45\textwidth}
    \ifpdf
      \input{./images/LSB_matching_ROC_005_IH.pdf_tex}
    \else
    \fi
  \label{fig:BOSS_COR_R05}  }
\hspace*{0.35cm}
  \subfloat[ROC curves for $R=1$.]{
  \def\svgwidth{0.45\textwidth}
  \hspace*{-0.55cm}
    \ifpdf
      \input{./images/LSB_matching_ROC_010_IH.pdf_tex}
    \else
    \fi
  \label{fig:BOSS_COR_R1}  }
  }
\caption{Numerical comparisons of detectors performance using BOSS database~\cite{BOSS2010}.}
\label{fig:num_results_BOSS}
\end{figure} 
\begin{figure}[!b]
\centerline{
\hspace*{-0.5cm}
  \subfloat[ROC curves for $R=0.25$.]{
  \hspace*{-0.55cm}
  \def\svgwidth{0.45\textwidth}
    \ifpdf
      \input{./images/LSB_matching_ROC_dresden_025_IH.pdf_tex}
    \else
    \fi
  \label{fig:Dresden_COR_R025}  }
\hspace*{0.35cm}
  \subfloat[ROC curves for $R=0.5$.]{
  \def\svgwidth{0.45\textwidth}
  \hspace*{-0.55cm}
    \ifpdf
      \input{./images/LSB_matching_ROC_dresden_05_IH.pdf_tex}
    \else
    \fi
  \label{fig:Dresden_COR_R05}  }
  }
\caption{Comparisons of detectors performance using Dresden database~\cite{DresdenDB2010}.}
\label{fig:num_results_Dresden}
\end{figure}
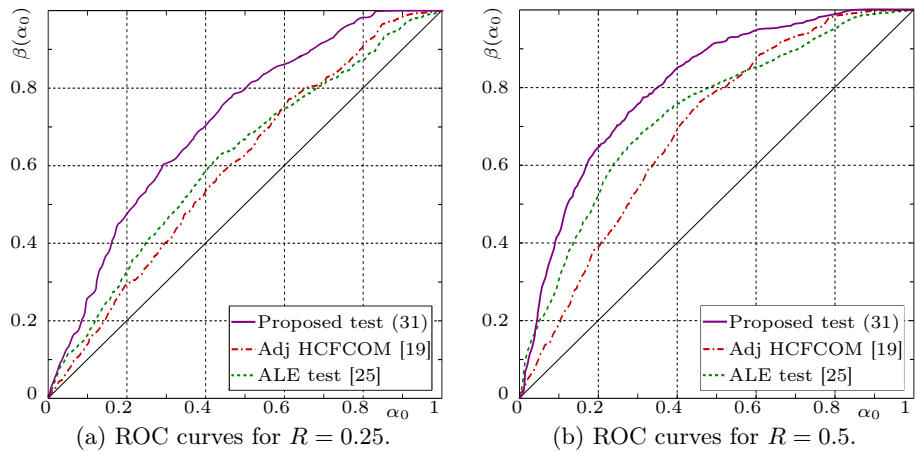 

One of the main motivations for this paper was to show that the hypothesis
testing theory can be applied in practice to design an efficient LSB
matching detector. This fact can only be shown by a numerical comparison
with state-of-the-art detectors on large image databases. The potential
competitors for LSB matching detection are not as numerous as for
LSB replacement. As briefly described in the introduction, the operational
context selected in this paper eliminates all prior-art detectors
based on machine learning. Almost every other detector found in the
literature is based on the image histogram. For the present comparison,
two histogram-based detectors, namely ALE~\cite{CoxDoerr07pm} and
the adjacency HCF COM~\cite{ker05pm} detector, were used due to
their high detection performance.\\
 Figure~\ref{fig:num_results_BOSS} shows the results obtained
with 10\,000 images from BOSSbase contest database~\cite{BOSS2010}.
Each hidden bit was drawn from a binomial distribution $\mathcal{B}(1,1/2)$.
The embedding rate was $R=0.5$ in Figure~\ref{fig:BOSS_COR_R05}
and $R=1$ in Figure~\ref{fig:BOSS_COR_R1}. Both figures show that
the proposed test achieves a better detection power for any prescribed
false-alarm probability.


Similarly, Figure~\ref{fig:num_results_Dresden} shows the results
obtained with the 1488 raw images from the `Dresden Image Database'~\cite{DresdenDB2010}. 
Prior to our experiments, each image was converted to an unprocessed
TIFF format (using dcraw) and only the red color channel was used.
The embedding rate was $R=0.25$ in Figure~\ref{fig:Dresden_COR_R025}
and $R=0.5$ in Figure~\ref{fig:Dresden_COR_R05}. The results presented
in Figures~\ref{fig:Dresden_COR_R025} and~\ref{fig:Dresden_COR_R05}
confirm that the proposed test has a better detection power for any
prescribed false-alarm probability. Moreover by changing the embedding
rate, the combined results of Figures~\ref{fig:num_results_BOSS}
and~\ref{fig:num_results_Dresden} show that the proposed test also
performs better than prior art for any $R$.

Note that, surprisingly, the detection power of the proposed test
is slightly higher for the BOSSbase database than for the Dresden
database for $R=0.5$, see Figure~\ref{fig:BOSS_COR_R05} and~\ref{fig:Dresden_COR_R05},
respectively, whereas the Dresden database images are bigger. This
phenomenon can be explained by the fact that the Dresden database
images are RAW images that have not being further processed. In contrast,
BOSSbase images have been downsampled, which may introduce correlations
between neighboring pixels that implicitly make the filtering estimator
more efficient.

\section{Conclusion and future works.}\label{sec:conclusion}

The first step to fill the gap between hypothesis testing theory and
steganalysis was recently proposed in~\cite{DMS2004,Cog2011IH,Zit2011IH}.
This paper extends this first step to the case of LSB matching. By
casting the problem of LSB matching steganalysis in the framework
of hypothesis testing theory, the most powerful likelihood ratio test
is designed. Then, a thorough statistical study permits analytical
calculations of its performance in terms of the false-alarm probability
and detection power. To apply this test in practice, unknown image
parameters have to be estimated. Based on a simple estimation of these
unknown parameters, a practical test is proposed.\\
 The relevance of the proposed approach is emphasized through
numerical experiments. Compared to two leading histogram-based detectors,
the proposed practical test achieves a better detection power.

However, the practical test presented in this paper relies on a simple
yet efficient filtered version of inspected media to estimate pixel
expectations and variances. In our future work, a more efficient model
should be used to increase the detection power. Lastly, a thorough
statistical study of the impact of this estimation on detection performance
is desirable to complete the present work.

\section{Acknowledgments.}\label{sec:ack}

The authors would like to thank Jessica Fridrich for her contributions and stimulating discussions.



\end{document}

%% file: images/LSB_match_m0_m1_IH.pdf_tex
\begingroup
  \makeatletter
  \providecommand\color[2][]{%
    \errmessage{(Inkscape) Color is used for the text in Inkscape, but the package 'color.sty' is not loaded}
    \renewcommand\color[2][]{}%
  }
  \providecommand\transparent[1]{%
    \errmessage{(Inkscape) Transparency is used (non-zero) for the text in Inkscape, but the package 'transparent.sty' is not loaded}
    \renewcommand\transparent[1]{}%
  }
  \providecommand\rotatebox[2]{#2}
  \ifx\svgwidth\undefined
    \setlength{\unitlength}{698.03129687pt}
  \else
    \setlength{\unitlength}{\svgwidth}
  \fi
  \global\let\svgwidth\undefined
  \makeatother
  \begin{picture}(1,0.8537998)%
    \put(0,0.04){\includegraphics[width=\unitlength]{./images/LSB_match_m0_m1_IH.pdf}}%
\begin{scriptsize}
    \put(0.03635667,0.01042835){\color[rgb]{0,0,0}\makebox(0,0)[b]{\smash{126}}}%
    \put(0.24929752,0.01042835){\color[rgb]{0,0,0}\makebox(0,0)[b]{\smash{127}}}%
    \put(0.48911484,0.01042835){\color[rgb]{0,0,0}\makebox(0,0)[b]{\smash{128}}}%
    \put(0.72893217,0.01042835){\color[rgb]{0,0,0}\makebox(0,0)[b]{\smash{129}}}%
    \put(0.83884083,0.01042835){\color[rgb]{0,0,0}\makebox(0,0)[b]{\smash{$\theta$}}}%
    \put(0.9487495,0.01042835){\color[rgb]{0,0,0}\makebox(0,0)[b]{\smash{130}}}%
    \put(0.00404777,0.05734919){\color[rgb]{0,0,0}\makebox(0,0)[rb]{\smash{-1.8}}}%
    \put(0.00404777,0.22869143){\color[rgb]{0,0,0}\makebox(0,0)[rb]{\smash{-1.4}}}%
    \put(0.00404777,0.41599329){\color[rgb]{0,0,0}\makebox(0,0)[rb]{\smash{-1}}}%
    \put(0.0004777,0.60329366){\color[rgb]{0,0,0}\makebox(0,0)[rb]{\smash{-0.6}}}%
    \put(0.00404777,0.70694459){\color[rgb]{0,0,0}\makebox(0,0)[rb]{\smash{$\mu_i(\theta)$}}}%
    \put(0.00404777,0.78059552){\color[rgb]{0,0,0}\makebox(0,0)[rb]{\smash{-0.2}}}%
    \put(0.65382576,0.56354988){\color[rgb]{0,0,0}\makebox(0,0)[lb]{\smash{$\calH_0$ simulation}}}%
    \put(0.65382576,0.5144068){\color[rgb]{0,0,0}\makebox(0,0)[lb]{\smash{$\calH_0$ theory}}}%
    \put(0.65382576,0.47175361){\color[rgb]{0,0,0}\makebox(0,0)[lb]{\smash{$\calH_1$ simulation}}}%
    \put(0.65382576,0.42761127){\color[rgb]{0,0,0}\makebox(0,0)[lb]{\smash{$\calH_1$ theory}}}%
\end{scriptsize}
  \end{picture}%
\endgroup

%% file: images/LSB_match_s0_s1_IH.pdf_tex
\begingroup
  \makeatletter
  \providecommand\color[2][]{%
    \errmessage{(Inkscape) Color is used for the text in Inkscape, but the package 'color.sty' is not loaded}
    \renewcommand\color[2][]{}%
  }
  \providecommand\transparent[1]{%
    \errmessage{(Inkscape) Transparency is used (non-zero) for the text in Inkscape, but the package 'transparent.sty' is not loaded}
    \renewcommand\transparent[1]{}%
  }
  \providecommand\rotatebox[2]{#2}
  \ifx\svgwidth\undefined
    \setlength{\unitlength}{694.87145313pt}
  \else
    \setlength{\unitlength}{\svgwidth}
  \fi
  \global\let\svgwidth\undefined
  \makeatother
  \begin{picture}(1,0.85731856)%
    \put(0,0.04){\includegraphics[width=\unitlength]{./images/LSB_match_s0_s1_IH.pdf}}%
\begin{scriptsize}
    \put(0.03197462,0.0104303){\color[rgb]{0,0,0}\makebox(0,0)[b]{\smash{126}}}%
    \put(0.24597473,0.0104303){\color[rgb]{0,0,0}\makebox(0,0)[b]{\smash{127}}}%
    \put(0.4868826,0.0104303){\color[rgb]{0,0,0}\makebox(0,0)[b]{\smash{128}}}%
    \put(0.72779047,0.0104303){\color[rgb]{0,0,0}\makebox(0,0)[b]{\smash{129}}}%
    \put(0.8382444,0.0104303){\color[rgb]{0,0,0}\makebox(0,0)[b]{\smash{$\theta$}}}%
    \put(0.94869834,0.0104303){\color[rgb]{0,0,0}\makebox(0,0)[b]{\smash{130}}}%
    \put(0.00060974,0.12506535){\color[rgb]{0,0,0}\makebox(0,0)[rb]{\smash{1}}}%
    \put(0.00060974,0.29231183){\color[rgb]{0,0,0}\makebox(0,0)[rb]{\smash{2}}}%
    \put(0.00060974,0.4595598){\color[rgb]{0,0,0}\makebox(0,0)[rb]{\smash{3}}}%
    \put(0.0060974,0.62680627){\color[rgb]{0,0,0}\makebox(0,0)[rb]{\smash{4}}}%
    \put(0.00060974,0.7043025){\color[rgb]{0,0,0}\makebox(0,0)[rb]{\smash{$\sigma_i^2(\theta)$}}}%
    \put(0.00060974,0.77405424){\color[rgb]{0,0,0}\makebox(0,0)[rb]{\smash{5}}}%
    \put(0.65382576,0.53854988){\color[rgb]{0,0,0}\makebox(0,0)[lb]{\smash{$\calH_0$ simulation}}}%
    \put(0.65382576,0.4894068){\color[rgb]{0,0,0}\makebox(0,0)[lb]{\smash{$\calH_0$ theory}}}%
    \put(0.65382576,0.44475361){\color[rgb]{0,0,0}\makebox(0,0)[lb]{\smash{$\calH_1$ simulation}}}%
    \put(0.65382576,0.40261127){\color[rgb]{0,0,0}\makebox(0,0)[lb]{\smash{$\calH_1$ theory}}}%
\end{scriptsize}
  \end{picture}%
\endgroup

%% file: images/LSB_match_ROC_simul_IH.pdf_tex
\begingroup
  \makeatletter
  \providecommand\color[2][]{%
    \errmessage{(Inkscape) Color is used for the text in Inkscape, but the package 'color.sty' is not loaded}
    \renewcommand\color[2][]{}%
  }
  \providecommand\transparent[1]{%
    \errmessage{(Inkscape) Transparency is used (non-zero) for the text in Inkscape, but the package 'transparent.sty' is not loaded}
    \renewcommand\transparent[1]{}%
  }
  \providecommand\rotatebox[2]{#2}
  \ifx\svgwidth\undefined
    \setlength{\unitlength}{613.02716494pt}
  \else
    \setlength{\unitlength}{\svgwidth}
  \fi
  \global\let\svgwidth\undefined
  \makeatother
  \begin{picture}(1,0.84342458)%
    \put(0,0.04){\includegraphics[width=\unitlength]{./images/LSB_match_ROC_simul_IH.pdf}}%
\begin{scriptsize}
    \put(0.03154836,0.02031907){\color[rgb]{0,0,0}\makebox(0,0)[b]{\smash{$10^{\textrm{-}3}$}}}%
    \put(0.35046799,0.02031907){\color[rgb]{0,0,0}\makebox(0,0)[b]{\smash{$10^{\textrm{-}2}$}}}%
    \put(0.66938893,0.02031907){\color[rgb]{0,0,0}\makebox(0,0)[b]{\smash{$10^{\textrm{-}1}$}}}%
    \put(0.83054519,0.02031907){\color[rgb]{0,0,0}\makebox(0,0)[b]{\smash{$\alpha_0$}}}%
    \put(0.96830857,0.02031907){\color[rgb]{0,0,0}\makebox(0,0)[b]{\smash{$1$}}}%
    \put(0.00785315,0.77682892){\color[rgb]{0,0,0}\makebox(0,0)[rb]{\smash{$1$}}}%
    \put(0,0.76){\color[rgb]{0,0,0}\rotatebox{90}{\makebox(0,0)[rb]{\smash{$\beta(\alpha_0)$}}}}%
    \put(0.01285265,0.57101456){\color[rgb]{0,0,0}\makebox(0,0)[rb]{\smash{$0.7$}}}%
    \put(0.04,0.75){\color[rgb]{0.15,0.15,0.8}\makebox(0,0)[lb]{\smash{$\theta{=}127.5$}}}%
    \put(0.65,0.5){\color[rgb]{0.15,0.7,0.15}\makebox(0,0)[lb]{\smash{$\theta{=}128$}}}%
    \put(0.01285265,0.23086107){\color[rgb]{0,0,0}\makebox(0,0)[rb]{\smash{$0.4$}}}%
    \put(0.5540264,0.27004588){\color[rgb]{0,0,0}\makebox(0,0)[lb]{\smash{$\beta(\alpha_0)$: theory}}}%
    \put(0.5540264,0.20902281){\color[rgb]{0,0,0}\makebox(0,0)[lb]{\smash{$\beta(\alpha_0)$: simulation}}}%
    \put(0.5540264,0.14800144){\color[rgb]{0,0,0}\makebox(0,0)[lb]{\smash{$\beta(\alpha_0)$: theory}}}%
    \put(0.5540264,0.08697837){\color[rgb]{0,0,0}\makebox(0,0)[lb]{\smash{$\beta(\alpha_0)$: simulation}}}%
\end{scriptsize}
  \end{picture}%
\endgroup

%% file: images/LSB_match_FAProb_tau_simul_IH.pdf_tex
\begingroup
  \makeatletter
  \providecommand\color[2][]{%
    \errmessage{(Inkscape) Color is used for the text in Inkscape, but the package 'color.sty' is not loaded}
    \renewcommand\color[2][]{}%
  }
  \providecommand\transparent[1]{%
    \errmessage{(Inkscape) Transparency is used (non-zero) for the text in Inkscape, but the package 'transparent.sty' is not loaded}
    \renewcommand\transparent[1]{}%
  }
  \providecommand\rotatebox[2]{#2}
  \ifx\svgwidth\undefined
    \setlength{\unitlength}{605.57265pt}
  \else
    \setlength{\unitlength}{\svgwidth}
  \fi
  \global\let\svgwidth\undefined
  \makeatother
  \begin{picture}(1,0.84423531)%
    \put(0,0.04){\includegraphics[width=\unitlength]{./images/LSB_match_FAProb_tau_simul_IH.pdf}}%
\begin{scriptsize}
    \put(0.08333533,0.02032226){\color[rgb]{0,0,0}\makebox(0,0)[b]{\smash{-1.7}}}%
    \put(0.22789272,0.02032226){\color[rgb]{0,0,0}\makebox(0,0)[b]{\smash{-1.6}}}%
    \put(0.3724501,0.02032226){\color[rgb]{0,0,0}\makebox(0,0)[b]{\smash{-1.5}}}%
    \put(0.51700881,0.02032226){\color[rgb]{0,0,0}\makebox(0,0)[b]{\smash{-1.4}}}%
    \put(0.6615662,0.02032226){\color[rgb]{0,0,0}\makebox(0,0)[b]{\smash{-1.3}}}%
    \put(0.80612491,0.02032226){\color[rgb]{0,0,0}\makebox(0,0)[b]{\smash{-1.2}}}%
    \put(0.985054519,0.02032226){\color[rgb]{0,0,0}\makebox(0,0)[b]{\smash{$\tau$ ($\times 10^2$)}}}%
    \put(0.00725373,0.05433692){\color[rgb]{0,0,0}\makebox(0,0)[rb]{\smash{0}}}%
    \put(0.00725373,0.19232979){\color[rgb]{0,0,0}\makebox(0,0)[rb]{\smash{0.2}}}%
    \put(0.00725373,0.34032267){\color[rgb]{0,0,0}\makebox(0,0)[rb]{\smash{0.4}}}%
    \put(0.00725373,0.48831555){\color[rgb]{0,0,0}\makebox(0,0)[rb]{\smash{0.6}}}%
    \put(0.00725373,0.63630842){\color[rgb]{0,0,0}\makebox(0,0)[rb]{\smash{0.8}}}%
    \put(-0.02,0.815){\color[rgb]{0,0,0}\rotatebox{90}{\makebox(0,0)[rb]{\smash{$\alpha_0(\tau)$}}}}%
    \put(0.61,0.26){\color[rgb]{0.15,0.15,0.8}\makebox(0,0)[lb]{\smash{$\theta{=}127.5$}}}%
    \put(0.30,0.37){\color[rgb]{0.15,0.7,0.15}\makebox(0,0)[lb]{\smash{$\theta{=}128$}}}%
    \put(0.26917285,0.691){\color[rgb]{0,0,0}\makebox(0,0)[lb]{\smash{$\alpha_0(\tau)$: theory}}}%
    \put(0.26917285,0.633){\color[rgb]{0,0,0}\makebox(0,0)[lb]{\smash{$\alpha_0(\tau)$: simulation}}}%
    \put(0.26917285,0.580){\color[rgb]{0,0,0}\makebox(0,0)[lb]{\smash{$\alpha_0(\tau)$: theory}}}%
    \put(0.26917285,0.52782396){\color[rgb]{0,0,0}\makebox(0,0)[lb]{\smash{$\alpha_0(\tau)$: simulation}}}%
\end{scriptsize}
  \end{picture}%
\endgroup

%% file: images/new_lin.pdf_tex
\begingroup
  \makeatletter
  \providecommand\color[2][]{%
    \errmessage{(Inkscape) Color is used for the text in Inkscape, but the package 'color.sty' is not loaded}
    \renewcommand\color[2][]{}%
  }
  \providecommand\transparent[1]{%
    \errmessage{(Inkscape) Transparency is used (non-zero) for the text in Inkscape, but the package 'transparent.sty' is not loaded}
    \renewcommand\transparent[1]{}%
  }
  \providecommand\rotatebox[2]{#2}
  \ifx\svgwidth\undefined
    \setlength{\unitlength}{364.68025pt}
  \else
    \setlength{\unitlength}{\svgwidth}
  \fi
  \global\let\svgwidth\undefined
  \makeatother
  \begin{picture}(1,0.78445926)%
    \put(0,0){\includegraphics[width=\unitlength]{./images/new_lin.pdf}}%
\begin{scriptsize}
    \put(0.3753041,0.52598287){\color[rgb]{0,0,0.83137255}\makebox(0,0)[b]{\smash{$R\!=\!0.1$}}}%
    \put(0.23300978,0.63654015){\color[rgb]{0,0.83137255,0}\makebox(0,0)[b]{\smash{$R\!=\!0.2$}}}%
    \put(0.13332493,0.73885891){\color[rgb]{0.83137255,0,0}\makebox(0,0)[b]{\smash{$R\!=\!0.4$}}}%
    \put(0.05160754,0.00063069){\color[rgb]{0,0,0}\makebox(0,0)[b]{\smash{0}}}%
    \put(0.23763352,0.00063069){\color[rgb]{0,0,0}\makebox(0,0)[b]{\smash{0.2}}}%
    \put(0.4280469,0.00063069){\color[rgb]{0,0,0}\makebox(0,0)[b]{\smash{0.4}}}%
    \put(0.61846028,0.00063069){\color[rgb]{0,0,0}\makebox(0,0)[b]{\smash{0.6}}}%
    \put(0.80887366,0.00063069){\color[rgb]{0,0,0}\makebox(0,0)[b]{\smash{0.8}}}%
    \put(0.99489964,0.00063069){\color[rgb]{0,0,0}\makebox(0,0)[b]{\smash{1}}}%
    \put(0.03695361,0.0272754){\color[rgb]{0,0,0}\makebox(0,0)[rb]{\smash{0}}}%
    \put(0.03695361,0.17306887){\color[rgb]{0,0,0}\makebox(0,0)[rb]{\smash{0.2}}}%
    \put(0.03695361,0.32324975){\color[rgb]{0,0,0}\makebox(0,0)[rb]{\smash{0.4}}}%
    \put(0.03695361,0.47343063){\color[rgb]{0,0,0}\makebox(0,0)[rb]{\smash{0.6}}}%
    \put(0.03695361,0.62361151){\color[rgb]{0,0,0}\makebox(0,0)[rb]{\smash{0.8}}}%
    \put(0.03695361,0.76501757){\color[rgb]{0,0,0}\makebox(0,0)[rb]{\smash{1}}}%
    \put(0.565,0.15250){\color[rgb]{0,0,0}\makebox(0,0)[lb]{\smash{Proposed test $\widetilde{\delta_2}$}}}%
    \put(0.565,0.0875){\color[rgb]{0,0,0}\makebox(0,0)[lb]{\smash{Clairvoyant detector}}}%
\end{scriptsize}
  \end{picture}%
\endgroup

%% file: images/LSB_matching_ROC_compalpha.pdf_tex
\begingroup
  \makeatletter
  \providecommand\color[2][]{%
    \errmessage{(Inkscape) Color is used for the text in Inkscape, but the package 'color.sty' is not loaded}
    \renewcommand\color[2][]{}%
  }
  \providecommand\transparent[1]{%
    \errmessage{(Inkscape) Transparency is used (non-zero) for the text in Inkscape, but the package 'transparent.sty' is not loaded}
    \renewcommand\transparent[1]{}%
  }
  \providecommand\rotatebox[2]{#2}
  \ifx\svgwidth\undefined
    \setlength{\unitlength}{363.08025pt}
  \else
    \setlength{\unitlength}{\svgwidth}
  \fi
  \global\let\svgwidth\undefined
  \makeatother
  \begin{picture}(1,0.80350943)%
    \put(0.02,0.01){\includegraphics[width=\unitlength]{./images/LSB_matching_ROC_compalpha.pdf}}%
\begin{scriptsize}
    \put(0.04742822,0.00063347){\color[rgb]{0,0,0}\makebox(0,0)[b]{\smash{0}}}%
    \put(0.23427397,0.00063347){\color[rgb]{0,0,0}\makebox(0,0)[b]{\smash{0.2}}}%
    \put(0.42552645,0.00063347){\color[rgb]{0,0,0}\makebox(0,0)[b]{\smash{0.4}}}%
    \put(0.61677894,0.00063347){\color[rgb]{0,0,0}\makebox(0,0)[b]{\smash{0.6}}}%
    \put(0.80803142,0.00063347){\color[rgb]{0,0,0}\makebox(0,0)[b]{\smash{0.8}}}%
    \put(0.91620638,-0.007507509){\color[rgb]{0,0,0}\makebox(0,0)[b]{\smash{$\alpha_0$}}}%
    \put(0.99487717,0.00063347){\color[rgb]{0,0,0}\makebox(0,0)[b]{\smash{1}}}%
    \put(0.03711645,0.0273956){\color[rgb]{0,0,0}\makebox(0,0)[rb]{\smash{0}}}%
    \put(0.03711645,0.1694248){\color[rgb]{0,0,0}\makebox(0,0)[rb]{\smash{0.2}}}%
    \put(0.03711645,0.32026749){\color[rgb]{0,0,0}\makebox(0,0)[rb]{\smash{0.4}}}%
    \put(0.03711645,0.47111017){\color[rgb]{0,0,0}\makebox(0,0)[rb]{\smash{0.6}}}%
    \put(0.03711645,0.62195286){\color[rgb]{0,0,0}\makebox(0,0)[rb]{\smash{0.8}}}%
    \put(0.00,0.75515477){\color[rgb]{0,0,0}\rotatebox{90}{\makebox(0,0)[rb]{\smash{$\beta_{\widehat{\delta}_1}$}}}}%
    \put(0.03711645,0.76898207){\color[rgb]{0,0,0}\makebox(0,0)[rb]{\smash{1}}}%
    \put(0.6875,0.26102285){\color[rgb]{0,0,0}\makebox(0,0)[lb]{\smash{$\alpha=0.25$}}}%
    \put(0.6875,0.20153187){\color[rgb]{0,0,0}\makebox(0,0)[lb]{\smash{$\alpha=0.5$}}}%
    \put(0.6875,0.14204089){\color[rgb]{0,0,0}\makebox(0,0)[lb]{\smash{$\alpha=0.75$}}}%
    \put(0.6875,0.08254991){\color[rgb]{0,0,0}\makebox(0,0)[lb]{\smash{$\alpha=1$}}}%
\end{scriptsize}
  \end{picture}%
\endgroup

%% file: images/LSB_matching_ROC_compLR_new.pdf_tex
\begingroup
  \makeatletter
  \providecommand\color[2][]{%
    \errmessage{(Inkscape) Color is used for the text in Inkscape, but the package 'color.sty' is not loaded}
    \renewcommand\color[2][]{}%
  }
  \providecommand\transparent[1]{%
    \errmessage{(Inkscape) Transparency is used (non-zero) for the text in Inkscape, but the package 'transparent.sty' is not loaded}
    \renewcommand\transparent[1]{}%
  }
  \providecommand\rotatebox[2]{#2}
  \ifx\svgwidth\undefined
    \setlength{\unitlength}{363.080252pt}
  \else
    \setlength{\unitlength}{\svgwidth}
  \fi
  \global\let\svgwidth\undefined
  \makeatother
  \begin{picture}(1,0.80350943)%
    \put(0.02,0.01){\includegraphics[width=\unitlength]{./images/LSB_matching_ROC_compLR_new.pdf}}%
\begin{scriptsize}
    \put(0.04742822,0.00063347){\color[rgb]{0,0,0}\makebox(0,0)[b]{\smash{0}}}%
    \put(0.23427396,0.00063347){\color[rgb]{0,0,0}\makebox(0,0)[b]{\smash{0.2}}}%
    \put(0.42552646,0.00063347){\color[rgb]{0,0,0}\makebox(0,0)[b]{\smash{0.4}}}%
    \put(0.61677895,0.00063347){\color[rgb]{0,0,0}\makebox(0,0)[b]{\smash{0.6}}}%
    \put(0.80803143,0.00063347){\color[rgb]{0,0,0}\makebox(0,0)[b]{\smash{0.8}}}%
    \put(0.91620638,-0.007507509){\color[rgb]{0,0,0}\makebox(0,0)[b]{\smash{$\alpha_0$}}}%
    \put(0.99487718,0.00063347){\color[rgb]{0,0,0}\makebox(0,0)[b]{\smash{1}}}%
    \put(0.03711645,0.0273956){\color[rgb]{0,0,0}\makebox(0,0)[rb]{\smash{0}}}%
    \put(0.03711645,0.1694248){\color[rgb]{0,0,0}\makebox(0,0)[rb]{\smash{0.2}}}%
    \put(0.03711645,0.32026749){\color[rgb]{0,0,0}\makebox(0,0)[rb]{\smash{0.4}}}%
    \put(0.03711645,0.47111017){\color[rgb]{0,0,0}\makebox(0,0)[rb]{\smash{0.6}}}%
    \put(0.03711645,0.62195286){\color[rgb]{0,0,0}\makebox(0,0)[rb]{\smash{0.8}}}%
    \put(0.03711645,0.76398206){\color[rgb]{0,0,0}\makebox(0,0)[rb]{\smash{1}}}%
    \put(0.00,0.75515477){\color[rgb]{0,0,0}\rotatebox{90}{\makebox(0,0)[rb]{\smash{$\beta_{\widehat{\delta}_1}$}}}}%
    \put(0.5825,0.173){\color[rgb]{0,0,0}\makebox(0,0)[lb]{\smash{Test with $\widehat{\Lambda}_{2}$~\refeq{eq:simplifiedGLR_test}}}}
    \put(0.58,0.09){\color[rgb]{0,0,0}\makebox(0,0)[lb]{\smash{Test with $\widehat{\Lambda}^\star_2$~\refeq{eq:match_log_GLR_r_simple}}}}
\end{scriptsize}
  \end{picture}%
\endgroup

%% file: images/LSB_match_pwr_vs_NB_R.pdf_tex
\begingroup
  \makeatletter
  \providecommand\color[2][]{%
    \errmessage{(Inkscape) Color is used for the text in Inkscape, but the package 'color.sty' is not loaded}
    \renewcommand\color[2][]{}%
  }
  \providecommand\transparent[1]{%
    \errmessage{(Inkscape) Transparency is used (non-zero) for the text in Inkscape, but the package 'transparent.sty' is not loaded}
    \renewcommand\transparent[1]{}%
  }
  \providecommand\rotatebox[2]{#2}
  \ifx\svgwidth\undefined
    \setlength{\unitlength}{367.5612003pt}
  \else
    \setlength{\unitlength}{\svgwidth}
  \fi
  \global\let\svgwidth\undefined
  \makeatother
  \begin{picture}(1,0.93421112)%
    \put(0,00){\includegraphics[width=\unitlength, height=4.25cm]{./images/LSB_match_pwr_vs_NB_R.pdf}}%
\begin{scriptsize}
    \put(0.04901385,-0.007507509){\color[rgb]{0,0,0}\makebox(0,0)[b]{\smash{$5$}}}%
    \put(0.35517674,-0.007507509){\color[rgb]{0,0,0}\makebox(0,0)[b]{\smash{$10$}}}%
    \put(0.67004349,-0.007507509){\color[rgb]{0,0,0}\makebox(0,0)[b]{\smash{$15$}}}%
    \put(0.81620638,-0.007507509){\color[rgb]{0,0,0}\makebox(0,0)[b]{\smash{$\log_2(N)$}}}%
    \put(0.96620638,-0.007507509){\color[rgb]{0,0,0}\makebox(0,0)[b]{\smash{$20$}}}%
    \put(0.045,0.03400611){\color[rgb]{0,0,0}\makebox(0,0)[rb]{\smash{$2.10^{\textrm{-}3}$}}}
    \put(0.045,0.20){\color[rgb]{0,0,0}\makebox(0,0)[rb]{\smash{$10^{\textrm{-}2}$}}}%
    \put(0.045,0.46){\color[rgb]{0,0,0}\makebox(0,0)[rb]{\smash{$10^{\textrm{-}1}$}}}%
    \put(-0.015,0.64515477){\color[rgb]{0,0,0}\rotatebox{90}{\makebox(0,0)[rb]{\smash{$\beta_{\widehat{\delta}_1}$}}}}%
    \put(0.045,0.73515477){\color[rgb]{0,0,0}\makebox(0,0)[rb]{\smash{$1$}}}%
    \put(0.560,0.12012227){\color[rgb]{0,0,0}\makebox(0,0)[lb]{\smash{Empirical results}}}%
    \put(0.560,0.06207614){\color[rgb]{0,0,0}\makebox(0,0)[lb]{\smash{Theoretical results}}}%
    \put(0.2725,0.6625){\color[rgb]{0,0,0}\makebox(0,0)[rb]{\smash{$\alpha_0\!=\!10^{\textrm{-}1}$}}}%
    \put(0.285,0.41){\color[rgb]{0,0,0}\makebox(0,0)[rb]{\smash{$\alpha_0\!=\!10^{\textrm{-}2}$}}}%
    \put(0.915,0.47){\color[rgb]{0,0,0}\makebox(0,0)[rb]{\smash{$\alpha_0\!=\!10^{\textrm{-}3}$}}}%
\end{scriptsize}
  \end{picture}%
\endgroup

%% file: images/LSB_matching_ROC_005_IH.pdf_tex
\begingroup
  \makeatletter
  \providecommand\color[2][]{%
    \errmessage{(Inkscape) Color is used for the text in Inkscape, but the package 'color.sty' is not loaded}
    \renewcommand\color[2][]{}%
  }
  \providecommand\transparent[1]{%
    \errmessage{(Inkscape) Transparency is used (non-zero) for the text in Inkscape, but the package 'transparent.sty' is not loaded}
    \renewcommand\transparent[1]{}%
  }
  \providecommand\rotatebox[2]{#2}
  \ifx\svgwidth\undefined
    \setlength{\unitlength}{542.42159929pt}
  \else
    \setlength{\unitlength}{\svgwidth}
  \fi
  \global\let\svgwidth\undefined
  \makeatother
  \begin{picture}(1,0.96039252)%
    \put(0,0){\includegraphics[width=\unitlength]{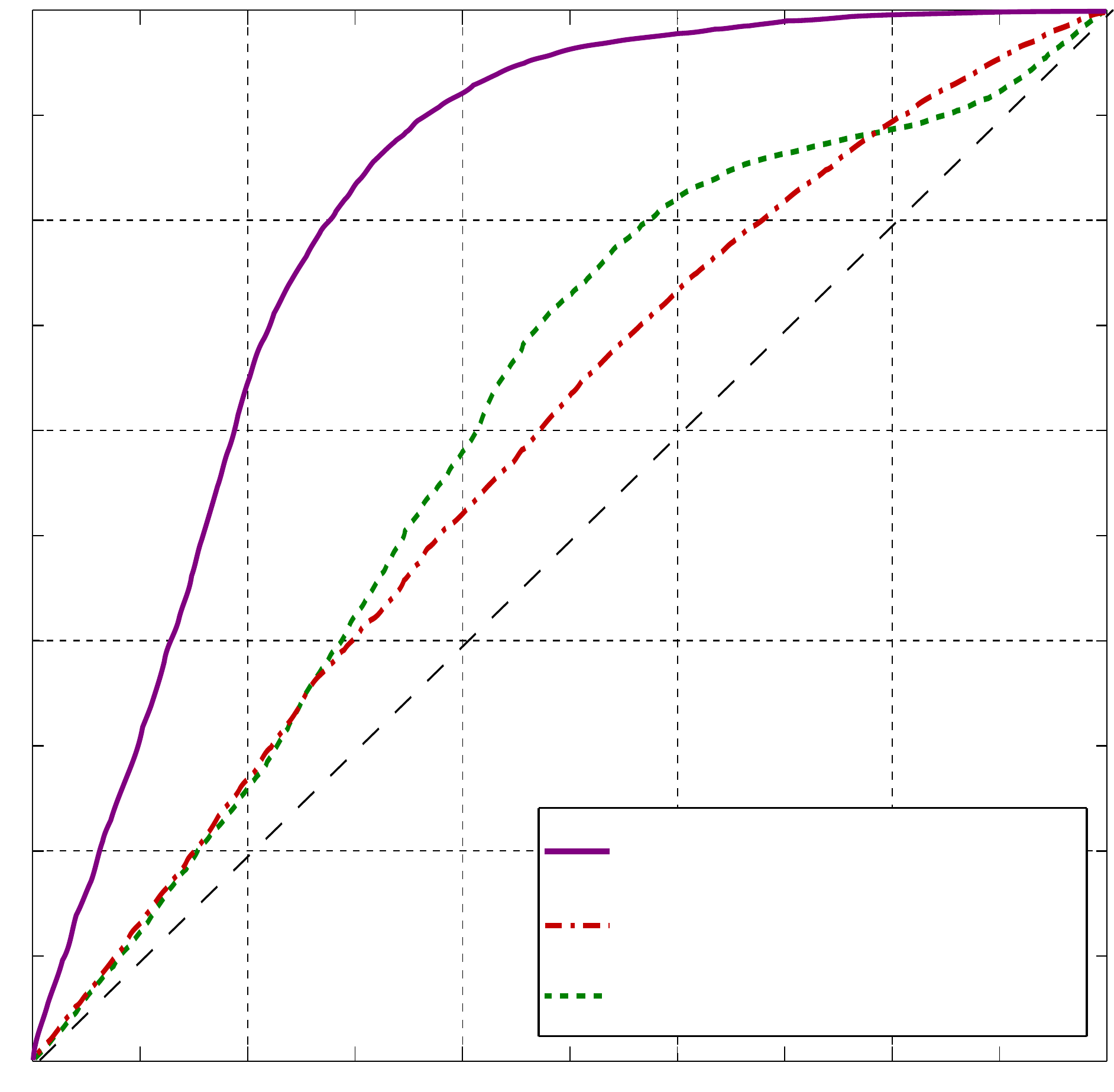}}%
\begin{scriptsize}
    \put(0.02028962,-0.008559363){\color[rgb]{0,0,0}\makebox(0,0)[b]{\smash{0}}}%
    \put(0.20733232,-0.008559363){\color[rgb]{0,0,0}\makebox(0,0)[b]{\smash{0.2}}}%
    \put(0.40027448,-0.008559363){\color[rgb]{0,0,0}\makebox(0,0)[b]{\smash{0.4}}}%
    \put(0.59321665,-0.008559363){\color[rgb]{0,0,0}\makebox(0,0)[b]{\smash{0.6}}}%
    \put(0.78615882,-0.008559363){\color[rgb]{0,0,0}\makebox(0,0)[b]{\smash{0.8}}}%
    \put(0.97910099,-0.008559363){\color[rgb]{0,0,0}\makebox(0,0)[b]{\smash{1}}}%
    \put(0.0078446,0.02578819){\color[rgb]{0,0,0}\makebox(0,0)[rb]{\smash{0}}}%
    \put(0.0078446,0.20184176){\color[rgb]{0,0,0}\makebox(0,0)[rb]{\smash{0.2}}}%
    \put(0.0078446,0.3906546){\color[rgb]{0,0,0}\makebox(0,0)[rb]{\smash{0.4}}}%
    \put(0.0078446,0.57946743){\color[rgb]{0,0,0}\makebox(0,0)[rb]{\smash{0.6}}}%
    \put(0.0078446,0.76828027){\color[rgb]{0,0,0}\makebox(0,0)[rb]{\smash{0.8}}}%
    \put(-0.01,0.965){\color[rgb]{0,0,0}\rotatebox{90}{\makebox(0,0)[rb]{\smash{$\beta(\alpha_0)$}}}}%
    \put(0.91,-0.0125){\color[rgb]{0,0,0}\rotatebox{00}{\makebox(0,0)[rb]{\smash{$\alpha_0$}}}}%
    \put(0.55,0.19609607){\color[rgb]{0,0,0}\makebox(0,0)[lb]{\smash{Proposed test~\refeq{eq:simplifiedGLR_test}}}}%
    \put(0.55,0.13013351){\color[rgb]{0,0,0}\makebox(0,0)[lb]{\smash{Adj HCFCOM~\cite{ker05pm}}}}%
    \put(0.55,0.06619239){\color[rgb]{0,0,0}\makebox(0,0)[lb]{\smash{ALE test~\cite{CoxDoerr07pm}}}}%
\end{scriptsize}
  \end{picture}%
\endgroup

%% file: images/LSB_matching_ROC_010_IH.pdf_tex
\begingroup
  \makeatletter
  \providecommand\color[2][]{%
    \errmessage{(Inkscape) Color is used for the text in Inkscape, but the package 'color.sty' is not loaded}
    \renewcommand\color[2][]{}%
  }
  \providecommand\transparent[1]{%
    \errmessage{(Inkscape) Transparency is used (non-zero) for the text in Inkscape, but the package 'transparent.sty' is not loaded}
    \renewcommand\transparent[1]{}%
  }
  \providecommand\rotatebox[2]{#2}
  \ifx\svgwidth\undefined
    \setlength{\unitlength}{539.90879633pt}
  \else
    \setlength{\unitlength}{\svgwidth}
  \fi
  \global\let\svgwidth\undefined
  \makeatother
  \begin{picture}(1,0.96495541)%
    \put(0,0){\includegraphics[width=\unitlength]{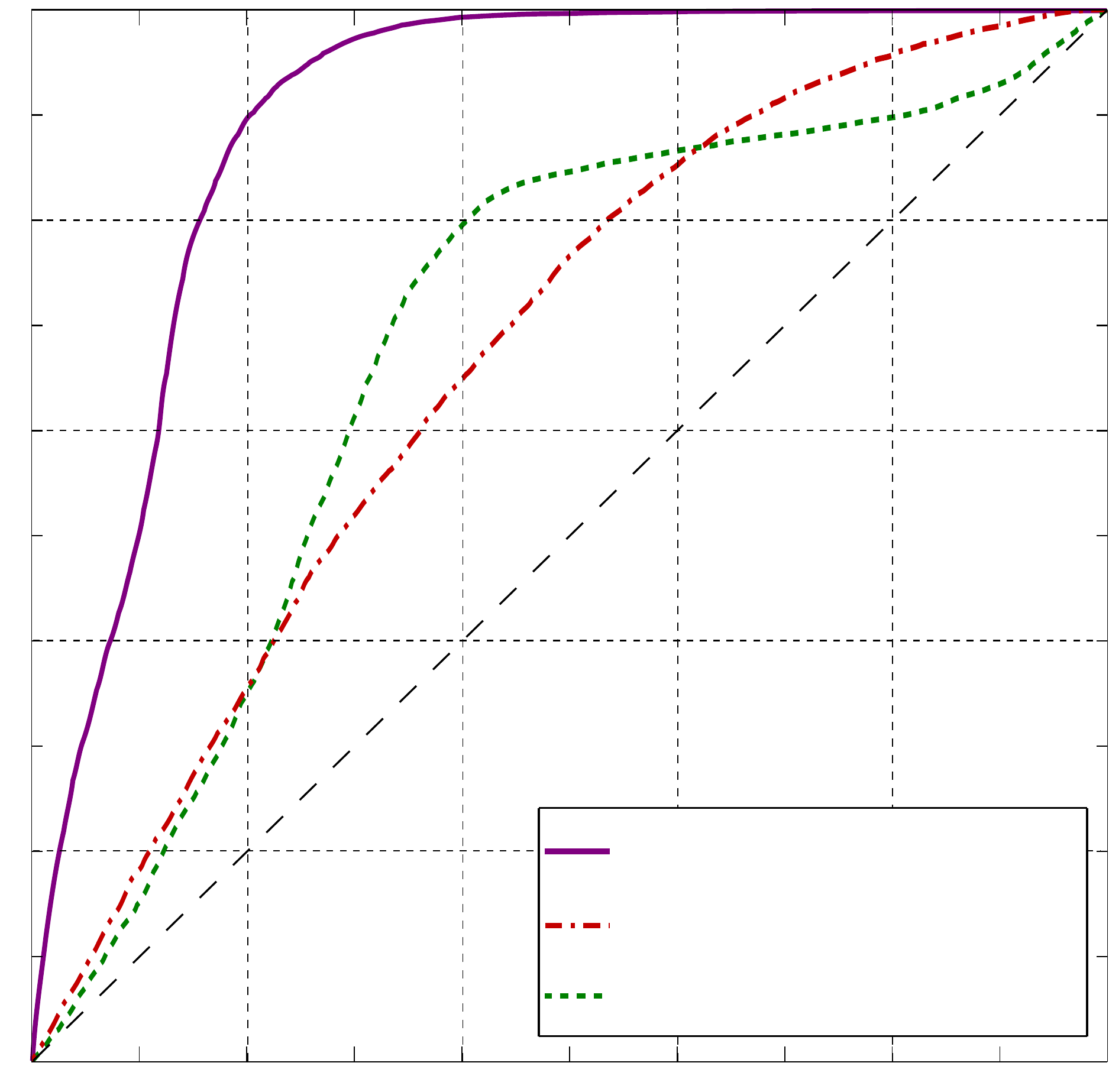}}%
\begin{scriptsize}
    \put(0.02045414,-0.00855964){\color[rgb]{0,0,0}\makebox(0,0)[b]{\smash{0}}}%
    \put(0.20836736,-0.00855964){\color[rgb]{0,0,0}\makebox(0,0)[b]{\smash{0.2}}}%
    \put(0.4022075,-0.00855964){\color[rgb]{0,0,0}\makebox(0,0)[b]{\smash{0.4}}}%
    \put(0.59604765,-0.00855964){\color[rgb]{0,0,0}\makebox(0,0)[b]{\smash{0.6}}}%
    \put(0.77988779,-0.00855964){\color[rgb]{0,0,0}\makebox(0,0)[b]{\smash{0.8}}}%
    \put(0.98372794,-0.00855964){\color[rgb]{0,0,0}\makebox(0,0)[b]{\smash{1}}}%
    \put(0.00796023,0.02141055){\color[rgb]{0,0,0}\makebox(0,0)[rb]{\smash{0}}}%
    \put(0.00796023,0.20280444){\color[rgb]{0,0,0}\makebox(0,0)[rb]{\smash{0.2}}}%
    \put(0.00796023,0.39249605){\color[rgb]{0,0,0}\makebox(0,0)[rb]{\smash{0.4}}}%
    \put(0.00796023,0.58218763){\color[rgb]{0,0,0}\makebox(0,0)[rb]{\smash{0.6}}}%
    \put(0.00796023,0.77187924){\color[rgb]{0,0,0}\makebox(0,0)[rb]{\smash{0.8}}}%
    \put(-0.01,0.965){\color[rgb]{0,0,0}\rotatebox{90}{\makebox(0,0)[rb]{\smash{$\beta(\alpha_0)$}}}}%
    \put(0.91,-0.0125){\color[rgb]{0,0,0}\rotatebox{00}{\makebox(0,0)[rb]{\smash{$\alpha_0$}}}}%
    \put(0.55,0.20308322){\color[rgb]{0,0,0}\makebox(0,0)[lb]{\smash{Proposed test~\refeq{eq:simplifiedGLR_test}}}}%
    \put(0.55,0.13681366){\color[rgb]{0,0,0}\makebox(0,0)[lb]{\smash{Adj HCFCOM~\cite{ker05pm}}}}%
    \put(0.55,0.07257496){\color[rgb]{0,0,0}\makebox(0,0)[lb]{\smash{ALE test~\cite{CoxDoerr07pm}}}}%
\end{scriptsize}
  \end{picture}%
\endgroup

%% file: images/LSB_matching_ROC_dresden_025_IH.pdf_tex
\begingroup
  \makeatletter
  \providecommand\color[2][]{%
    \errmessage{(Inkscape) Color is used for the text in Inkscape, but the package 'color.sty' is not loaded}
    \renewcommand\color[2][]{}%
  }
  \providecommand\transparent[1]{%
    \errmessage{(Inkscape) Transparency is used (non-zero) for the text in Inkscape, but the package 'transparent.sty' is not loaded}
    \renewcommand\transparent[1]{}%
  }
  \providecommand\rotatebox[2]{#2}
  \ifx\svgwidth\undefined
    \setlength{\unitlength}{543.19999427pt}
  \else
    \setlength{\unitlength}{\svgwidth}
  \fi
  \global\let\svgwidth\undefined
  \makeatother
  \begin{picture}(1,0.97785507)%
    \put(0,0){\includegraphics[width=\unitlength]{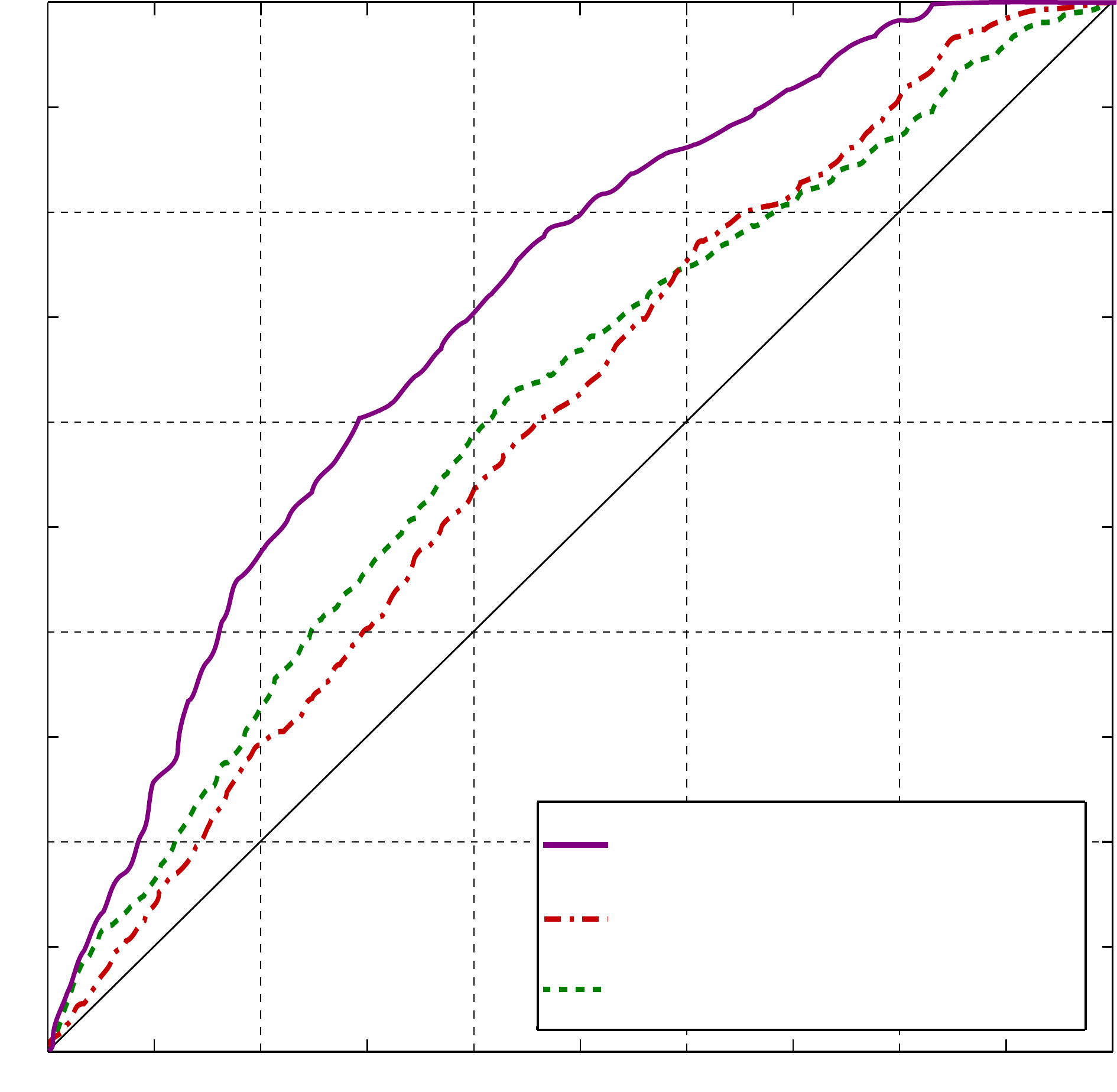}}%
\begin{scriptsize}
    \put(0.03737253,0.0007906){\color[rgb]{0,0,0}\makebox(0,0)[b]{\smash{0}}}%
    \put(0.2139989,0.0007906){\color[rgb]{0,0,0}\makebox(0,0)[b]{\smash{0.2}}}%
    \put(0.40502684,0.0007906){\color[rgb]{0,0,0}\makebox(0,0)[b]{\smash{0.4}}}%
    \put(0.59605477,0.0007906){\color[rgb]{0,0,0}\makebox(0,0)[b]{\smash{0.6}}}%
    \put(0.7870827,0.0007906){\color[rgb]{0,0,0}\makebox(0,0)[b]{\smash{0.8}}}%
    \put(0.97811064,0.0007906){\color[rgb]{0,0,0}\makebox(0,0)[b]{\smash{1}}}%
    \put(0.01707287,0.03419094){\color[rgb]{0,0,0}\makebox(0,0)[rb]{\smash{0}}}%
    \put(0.01707287,0.21144979){\color[rgb]{0,0,0}\makebox(0,0)[rb]{\smash{0.2}}}%
    \put(0.01707287,0.39970826){\color[rgb]{0,0,0}\makebox(0,0)[rb]{\smash{0.4}}}%
    \put(0.01707287,0.58796673){\color[rgb]{0,0,0}\makebox(0,0)[rb]{\smash{0.6}}}%
    \put(0.01707287,0.77622519){\color[rgb]{0,0,0}\makebox(0,0)[rb]{\smash{0.8}}}%
    \put(-0.01,0.985){\color[rgb]{0,0,0}\rotatebox{90}{\makebox(0,0)[rb]{\smash{$\beta(\alpha_0)$}}}}%
    \put(0.92,-0.009){\color[rgb]{0,0,0}\rotatebox{00}{\makebox(0,0)[rb]{\smash{$\alpha_0$}}}}%
    \put(0.55,0.20635361){\color[rgb]{0,0,0}\makebox(0,0)[lb]{\smash{Proposed test~\refeq{eq:simplifiedGLR_test}}}}%
    \put(0.55,0.14048558){\color[rgb]{0,0,0}\makebox(0,0)[lb]{\smash{Adj HCFCOM~\cite{ker05pm}}}}%
    \put(0.55,0.0766361){\color[rgb]{0,0,0}\makebox(0,0)[lb]{\smash{ALE test~\cite{CoxDoerr07pm}}}}%
\end{scriptsize}
  \end{picture}%
\endgroup

%% file: images/LSB_matching_ROC_dresden_05_IH.pdf_tex
\begingroup
  \makeatletter
  \providecommand\color[2][]{%
    \errmessage{(Inkscape) Color is used for the text in Inkscape, but the package 'color.sty' is not loaded}
    \renewcommand\color[2][]{}%
  }
  \providecommand\transparent[1]{%
    \errmessage{(Inkscape) Transparency is used (non-zero) for the text in Inkscape, but the package 'transparent.sty' is not loaded}
    \renewcommand\transparent[1]{}%
  }
  \providecommand\rotatebox[2]{#2}
  \ifx\svgwidth\undefined
    \setlength{\unitlength}{543.19999427pt}
  \else
    \setlength{\unitlength}{\svgwidth}
  \fi
  \global\let\svgwidth\undefined
  \makeatother
  \begin{picture}(1,0.97785507)%
    \put(0,0){\includegraphics[width=\unitlength]{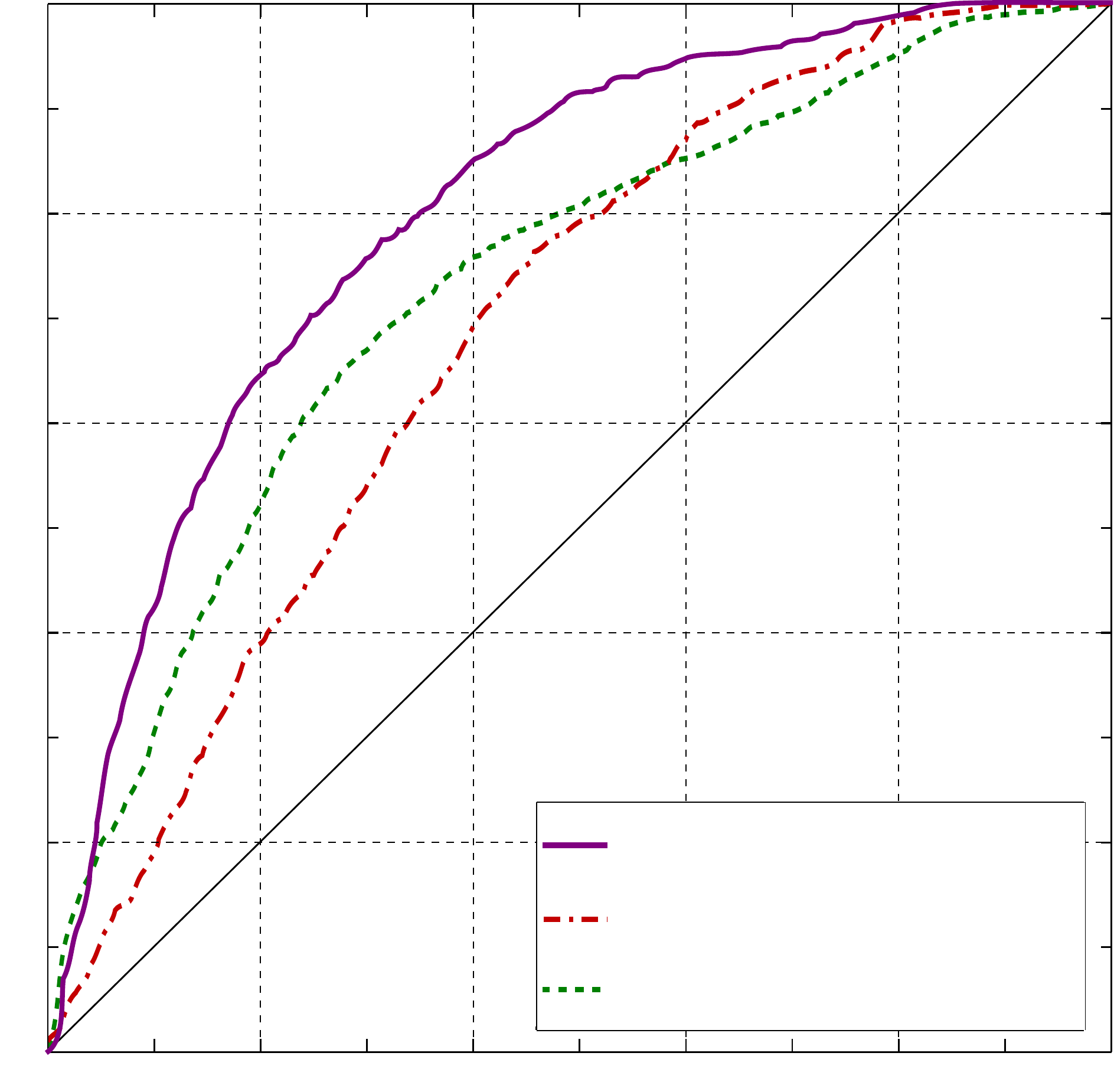}}%
\begin{scriptsize}
    \put(0.03737253,0.0007906){\color[rgb]{0,0,0}\makebox(0,0)[b]{\smash{0}}}%
    \put(0.2139989,0.0007906){\color[rgb]{0,0,0}\makebox(0,0)[b]{\smash{0.2}}}%
    \put(0.40502684,0.0007906){\color[rgb]{0,0,0}\makebox(0,0)[b]{\smash{0.4}}}%
    \put(0.59605477,0.0007906){\color[rgb]{0,0,0}\makebox(0,0)[b]{\smash{0.6}}}%
    \put(0.7870827,0.0007906){\color[rgb]{0,0,0}\makebox(0,0)[b]{\smash{0.8}}}%
    \put(0.97811064,0.0007906){\color[rgb]{0,0,0}\makebox(0,0)[b]{\smash{1}}}%
    \put(0.01707287,0.03419094){\color[rgb]{0,0,0}\makebox(0,0)[rb]{\smash{0}}}%
    \put(0.01707287,0.21144979){\color[rgb]{0,0,0}\makebox(0,0)[rb]{\smash{0.2}}}%
    \put(0.01707287,0.39970826){\color[rgb]{0,0,0}\makebox(0,0)[rb]{\smash{0.4}}}%
    \put(0.01707287,0.58796673){\color[rgb]{0,0,0}\makebox(0,0)[rb]{\smash{0.6}}}%
    \put(0.01707287,0.77622519){\color[rgb]{0,0,0}\makebox(0,0)[rb]{\smash{0.8}}}%
    \put(-0.01,0.985){\color[rgb]{0,0,0}\rotatebox{90}{\makebox(0,0)[rb]{\smash{$\beta(\alpha_0)$}}}}%
    \put(0.92,-0.009){\color[rgb]{0,0,0}\rotatebox{00}{\makebox(0,0)[rb]{\smash{$\alpha_0$}}}}%
    \put(0.55,0.20635361){\color[rgb]{0,0,0}\makebox(0,0)[lb]{\smash{Proposed test~\refeq{eq:simplifiedGLR_test}}}}%
    \put(0.55,0.14048558){\color[rgb]{0,0,0}\makebox(0,0)[lb]{\smash{Adj HCFCOM~\cite{ker05pm}}}}%
    \put(0.55,0.0766361){\color[rgb]{0,0,0}\makebox(0,0)[lb]{\smash{ALE test~\cite{CoxDoerr07pm}}}}%
\end{scriptsize}
  \end{picture}%
\endgroup